\documentclass[journal]{IEEEtran}

\usepackage{cite}

\usepackage[cmex10]{amsmath}
\usepackage{url}

\usepackage{amsopn,amsthm,amssymb,upref,verbatim}
\usepackage{pgfplots}
\pgfplotsset{compat=1.16}
\tikzset{vertex/.style={ draw , circle , fill , inner sep=0em , minimum size=0.3em}}
\tikzset{empty/.style={inner sep=0em, outer sep=0em, minimum size=0em}}

\makeatletter
\def\ps@IEEEtitlepagestyle{
  \def\@oddfoot{\mycopyrightnotice}
  \def\@evenfoot{}
}
\def\mycopyrightnotice{
  {\footnotesize
  \begin{minipage}{\textwidth}
  \centering
  Copyright~\copyright~2017 IEEE. Personal use of this material is permitted. However, permission to use this  \\ 
  material for any other purposes must be obtained from the IEEE by sending a request to pubs-permissions@ieee.org.
  \end{minipage}
  }
}

\usepackage{mathtools} 
\usepackage{bbm} 

\let\le\leqslant
\let\ge\geqslant

\newcommand{\mat}[1]{#1}

\renewcommand{\mod}{\mathop{\mathrm{mod}}}
\newcommand{\floor}[1]{{\left\lfloor{#1}\right\rfloor}} 
\newcommand{\ceil}[1]{{\left\lceil{#1}\right\rceil}} 
\newcommand{\abs}[1]{{\left\lvert#1\right\rvert}} 
\newcommand{\rbr}[1]{\left(#1\right)}
\newcommand{\ones}[1]{\mathrm{j}_{#1}}

\newcommand{\Mat}{\mathcal{M}} 
\newcommand{\F}{\mathbb{F}}

\newcommand{\HX}{H_{\mathrm{X}}}
\newcommand{\HZ}{H_{\mathrm{Z}}}

\newcommand{\CX}{\mathcal{C}_{\mathrm{X}}}
\newcommand{\CZ}{\mathcal{C}_{\mathrm{Z}}}

\newcommand{\dX}{d_{\mathrm{X}}}
\newcommand{\dZ}{d_{\mathrm{Z}}}

\newcommand{\T}{\mathrm{T}}

\newcommand{\id}{\mathrm{id}}

\newcommand{\HP}{\mathrm{HP}}
\newcommand{\LP}{\mathrm{LP}}

\newcommand{\C}{\mathbb{C}}

\newcommand{\Z}{\mathbb{Z}}
\newcommand{\bv}{\mathbbm{b}}
\newcommand{\B}{\mathbb{B}}

\newcommand{\zm}{\mathbf{0}}

\newcommand{\cA}{\mathcal{A}}
\newcommand{\cB}{\mathcal{B}}
\newcommand{\cC}{\mathcal{C}}

\newcommand{\cQ}{\mathcal{Q}}

\newcommand{\cT}{\mathcal{T}}

\DeclareMathOperator{\rk}{rk}
\DeclareMathOperator{\im}{im}

\DeclareMathOperator{\ch}{char}
\DeclareMathOperator{\spn}{span}
\DeclareMathOperator{\perm}{perm}

\DeclareMathOperator{\polylog}{\mathrm{polylog}}

\newcommand{\ket}[1]{\lvert{#1}\rangle}

\newcommand{\eps}{\varepsilon}
\renewcommand{\phi}{\varphi}

\theoremstyle{plain}
\newtheorem{lemma}{Lemma}
\newtheorem*{lemma*}{Lemma}
\newtheorem{proposition}{Proposition}
\newtheorem{theorem}{Theorem}
\newtheorem{corollary}{Corollary}

\theoremstyle{definition}

\theoremstyle{remark}
\newtheorem{remark}{Remark}
\newtheorem{example}{Example}
\newtheorem*{remark*}{Remark}

\makeatletter
\renewcommand*\env@matrix[1][*\c@MaxMatrixCols c]{%
  \hskip -\arraycolsep
  \let\@ifnextchar\new@ifnextchar
  \array{#1}}
\makeatother

\begin{document}


%
\title{Quantum LDPC Codes with Almost Linear Minimum Distance}

%

\author{Pavel~Panteleev and Gleb~Kalachev\thanks{Pavel~Panteleev and Gleb~Kalachev are with the Faculty of Mechanics and Mathematics, Moscow State University, Moscow, Russia.%
}}

%

\maketitle

\begin{abstract}
We give a~construction of quantum  LDPC codes  of dimension $\Theta(\log N)$ and distance $\Theta(N/\log N)$ as the~code length~$N\to\infty$. Using a product of chain complexes this construction also provides a family of quantum LDPC codes of distance~$\Omega(N^{1-\alpha/2}/\log N)$ and dimension~$\Omega(N^\alpha \log N)$, where $0 \le \alpha < 1$. We also introduce and study a new operation called lifted product, which naturally generalizes the product operations for quantum codes and chain complexes. Moreover, as a~simple byproduct of our results on quantum codes, we obtain a~new result on classical codes. We show that for any fixed $R < 1$ there exists an~asymptotically good family of classical quasi-cyclic LDPC codes of rate at least $R$ with, in some sense, optimal circulant size~$\Omega(N/\log N)$ as the~code length~$N\to\infty$.
\end{abstract}

\section{Introduction}

\IEEEPARstart{C}{lassical low-density} parity-check (LDPC) codes \cite{gallager1963} are a~very important class of linear codes widely used in theory and practise. The~definitive property of a family of LDPC codes is that there exists some constant $w$ such that for any code from this family both the~row and the~column weights of its parity-check matrix are bounded above by~$w$. The~theoretical importance of LDPC codes stems mostly from the fact that they contain asymptotically good codes of any positive rate with a~linear time decoding that can attain the~Shannon capacity~\cite{Sipser&Spielman:1996, Barg&Zemor:2002}. Their quantum analogs, called quantum LDPC (QLDPC) codes (see~\cite{Babar:2015} for a~good review), may play a~very important role in design of future fault-tolerant quantum computers~\cite{Gottesman:2014, Fawzi:2018}. However, it~is still unknown whether there exists an asymptotically good family of QLDPC codes with a~positive rate. More dramatically, to~the~best of our knowledge, there are even no such examples of constant dimension and linear distance, while in the classical case we have the~repetition code as a~trivial example. 

Up until very recently, the~minimum distance of all known examples of QLDPC codes~\cite{Kitaev:2002, Freedman:2002:best-code, Tillich&Zemor:2009, Evra:2020, Kaufman:2021} was bounded above by $O(N^{1/2} \log^\alpha N$) for some $\alpha\ge 0$ as the~code length~${N\to\infty}$. In~\cite{Hastings:2021:fiber}
it was shown that there exists a~family of QLDPC codes of distance and dimension bounded below by $\Omega(N^{3/5}/\polylog N)$. The QLDPC codes from the~all above-mentioned papers belong to a~wide class of quantum codes called CSS codes~\cite{CSS:1996, CSS2:1996}. A~CSS code $\cQ$ of~dimension~$K$ is defined by a pair of classical linear codes $\CZ,\CX\subseteq \F_2^N$ such that $\CX^\perp\subseteq \CZ$, and $K = \dim \CZ/\CX^\perp$. Its~minimum distance $d$ is defined as $\min(\dZ,\dX)$, where $\dZ$ and $\dX$ are the minimal Hamming weights of the~vectors from $\CZ\setminus\CX^\perp$ and $\CX\setminus\CZ^\perp$, respectively. In this case we often say that $\cQ$ is an~$[[N,K,d]]$~code, or, if we want to be more precise, an~$[[N,K, d_\mathrm{Z},d_\mathrm{X}]]$~code. The code $\CZ$ is usually represented by a parity-check matrix~$\HX$, and the code~$\CX$ by a parity-check matrix $\HZ$, and $\CX^\perp\subseteq \CZ$ implies that $\HX \HZ^\T = \zm$.

The approach used in~\cite{Evra:2020, Kaufman:2021, Hastings:2021:fiber} is, first, to construct a~quantum code~$\cQ$ with $\dZ \gg \dX$, $\dZ \dX > \Omega(N)$, and then to apply the~homological product~\cite{Hastings:2017, Pryadko:2019, Kaufman:2021} of~the quantum code~$\cQ$ with a~classical code $\cC$ of minimal distance~$d\approx \dZ/\dX$ in order to obtain a~new quantum code $\cQ'$ of distance $\min(\dZ,d\cdot\dX)$. In~\cite{Hastings:2021:fiber} this ``distance balancing'' procedure was applied to a~family of codes (called fiber \mbox{bundle} codes) with parameters ${\dZ=\Omega(N^{3/4}/\polylog N)}$, ${\dX = \Omega(N^{1/2})}$, and ${K=\Theta(N^{1/2})}$. We~should note that this particular family of fiber bundle codes coincides\footnote{The definition of these codes in~\cite{Hastings:2021:fiber} is given in terms of chain complexes, while in~\cite{Panteleev&Kalachev:2019} these codes are defined by parity-check matrices $\HX$ and $\HZ$.} with an earlier proposed~\cite{Panteleev&Kalachev:2019} family of quasi-cyclic GHP codes\footnote{In the~current paper we further generalize these codes and call them \emph{lifted product codes}.}, defined by some quasi-cyclic matrix $A$ of circulant size $\ell$ and the~polynomial ${b=1+x}$, which is a parity polynomial of the cyclic repetition code of length~$\ell$. The~parity-check matrices $\HX$, $\HZ$ for such codes are binary block matrices that look as follows:
\begin{equation}\label{eq:qc-ghp}
\begin{split}
    \HX &=
    \left[
    \begin{array}{@{}c|c@{}}
    \begin{matrix}
        A_{11} & \dots & A_{1n} \\
        \vdots & \ddots & \vdots \\
        A_{m1} & \dots & A_{mn}
    \end{matrix}
    &
    \begin{matrix}
        B & \dots & \zm \\
        \vdots & \ddots & \vdots \\
        \zm & \dots & B
    \end{matrix}
    \end{array}
\right]\!;\\
    \HZ &=
    \left[
    \begin{array}{@{}c|c@{}}
    \begin{matrix}
        B^\T & \dots & \zm \\
        \vdots & \ddots & \vdots \\
        \zm & \dots & B^\T
    \end{matrix}
    &
    \begin{matrix}
        A^\T_{11} & \dots & A^\T_{m1} \\
        \vdots & \ddots & \vdots \\
        A^\T_{1n} & \dots & A^\T_{mn}
    \end{matrix}
    \end{array}
\right]\!;
\end{split}
\end{equation}
where each $A_{ij}$ is an~$\ell\times\ell$-circulant matrix (see Appendix~\ref{sc:circulants}), and $B$ is the~$\ell\times\ell$ circulant matrix that is the parity-check matrix for the cyclic code with the~parity polynomial $b$. Since we want to obtain low-density matrices, the circulants in the~above block matrices should be as sparse as possible. This is the reason why in all the~examples of such codes in~\cite{Panteleev&Kalachev:2019} the~matrices $A_{i j}$ are circulants of weight~$1$, i.e., permutation matrices of some cyclic shifts modulo~$\ell$.

In the~terminology of~\cite{Hastings:2021:fiber}, the polynomial~$b$ corresponds to the~fiber, and the~matrix~$A$ to the~parity-check matrix of the~base with twists. In~\cite{Panteleev&Kalachev:2019} this class of codes was studied in the case of arbitrary parity polynomial $b$, and in the case of odd~$\ell$ a~formula for the dimension of such codes was given. Moreover, several examples of these codes were constructed, and one of them was shown to outperform under the~BP-OSD decoder (also proposed in~\cite{Panteleev&Kalachev:2019}) a~relatively large surface
code decoded by a near-optimal decoder from~\cite{Bravyi:2014}.

In this paper we show that if we carefully choose a~low-density~quasi-cyclic matrix $A$ and use $b=1+x$, then the corresponding GHP code has distance $\Theta(N/\log N)$ and dimension $\Theta(\log N)$ as the code length $N\to\infty$. This gives us our first main result.
\begin{theorem}\label{th:main1}
There exists a family of QLDPC codes of \mbox{dimension} $\Theta(\log N)$ and distance~$\Theta(N/\log N)$ as the~code length $N\to\infty$.
\end{theorem}
The main technical tool in the proof of the~above theorem is expander codes~\cite{Tanner1981, Sipser&Spielman:1996,Zemor:2001}. Such codes are defined by a~graph~$G$ and a~small linear code $\cC_0$. In order to obtain a~good expander code, the~second largest eigenvalue of the~adjacency matrix of~$G$ should be sufficiently small. A~graph\footnote{Formally, we should rather talk about an infinite family of graphs.} that satisfies this condition is called an~\emph{expander graph} (see~\cite{Hoory:2006} for a~good survey). An~important result, we rely on in our proof, is Theorem~1.2 from~\cite{Agarwal:2019}, which gives us a~way to construct quasi-cyclic matrices $A$ with the desired properties of very large circulant size $\ell = \Omega(N/\log N)$. Using this result, we first construct a~large expander graph~$G$ with a~quasi-cyclic adjacency matrix of circulant size $\ell$ from a~small expander graph; then we apply to $G$ the~expander code construction to obtain a~code $\cT(G,\cC_0)$ and define $A$ as its parity-check matrix. Since the~adjacency matrix of $G$ is quasi-cyclic, it is possible to define $\cT(G,\cC_0)$ in such a~way that its parity-check matrix $A$ is also quasi-cyclic of circulant size $\ell$.       

As a~byproduct of the proof of Theorem~\ref{th:main1} we also obtain (Corollary~\ref{cor:QC-dist}) that there exists a~family of classical quasi-cyclic LDPC codes with distance $\Theta(N)$ and circulant size $\Omega(N/\log N)$. Using the~well-known upper bound~\cite{Smarandache:2012} on the~minimal distance of quasi-cyclic LDPC codes we show that, in some sense, this circulant size is optimal.

Though the distance of the obtained quantum codes is almost linear as $N\to\infty$, their dimension is only $\Theta(\log N)$. In fact, the dimension~can be easily increased by a~moderate reduction of the code distance. The~idea is somewhat similar to the~mentioned above ``distance balancing'' procedure, but instead of the~code distance, we increase the~code dimension. As it was shown\footnote{Note that a~similar bound on the~distance was obtained earlier in~\cite[Theorem~1]{Pryadko:2019} in the language of chain complexes.} in~\cite[Theorem~2.3]{Evra:2020}, if we have a~quantum $[[N,K,d_Z, d_X]]$~code $\cQ$ and a~classical $[n,k,d]$~code $\cC$, then we can obtain the~quantum $[[N',k K, d'_\mathrm{Z}, d'_\mathrm{X}]]$~code $\cQ\otimes\cC$ called the~\emph{homological product} of $\cQ$ and $\cC$ such that:
\[ N' \le 2nN,\quad d'_\mathrm{Z} \ge d\cdot d_\mathrm{Z},\quad d'_\mathrm{X} \ge \dX. \]
Now if we consider the~quantum code $(\cQ\otimes \cC)^*$, where we change the roles of codes $\CZ$ and $\CX$ in $\cQ\otimes\cC$, and again apply the~homological product with~$\cC$, then we get the~$[[N'',k^2 K, d''_\mathrm{Z},d''_\mathrm{X}]]$~code\footnote{It is not hard to see that this construction is equivalent to the homological product of a quantum code and a~hypergraph product code defined by $\cC$.} $(\cQ\otimes \cC)^*\otimes \cC$ such that:
\begin{equation}\label{eq:hom-prod}
N'' \le 4n^2N,\quad d''_\mathrm{Z} \ge d\cdot d_\mathrm{Z},\quad d''_\mathrm{X} \ge d\cdot\dX.
\end{equation}
Therefore in order to obtain codes of large dimension out of the~constructed in this work codes of dimension $\Theta(\log N)$ and distance $\Theta(N/\log N)$ it remains to let $\cC$ be from a~family\footnote{As we already mentioned before, asymptotically good classical LDPC codes of non-vanishing rate do exist~\cite{gallager1963}.} of classical LDPC $[n,k,d]$ codes such that ${k = \Theta(n)}$, ${d = \Theta(n)}$, and $n = \Theta(N^{\frac{\alpha}{2(1-\alpha)}})$ as $N\to\infty$, where $\alpha > 0$. Indeed, we can easily check using (\ref{eq:hom-prod}) that as the end result we obtain the~quantum $[[N'',K'',d'']]$~code such that:
\begin{align*}
    N'' &= O(n^2N) =  O(N^{\frac{1}{1-\alpha}}), \text{ and hence } N = \Omega\!\rbr{(N'')^{1-\alpha}}\!;\\
    K'' &= k^2 K = \Theta(N^{\frac{\alpha}{1-\alpha}}\log N) = \Omega((N'')^\alpha\log N'');\\
        d'' &= \Omega(d\cdot N/\log N) = \Omega(n\cdot N/\log N) \\
        &= \Omega\!\rbr{(N'')^{\alpha/2}\cdot N/\log N} = \Omega\!\rbr{(N'')^{1 -\alpha/2}\log N''}\!.
\end{align*}
We should emphasize that  all the codes involved in the above construction have low-density parity-check matrices. Hence the obtained quantum codes are QLDPC codes, and we get the following result.
\begin{theorem}\label{th:main2}
For every $\alpha$ such that $0 \le \alpha < 1$ there exists a family of QLDPC codes of dimension~$\Omega(N^\alpha \log N)$ and distance~$\Omega(N^{1-\alpha/2}/\log N)$ as the~code length $N\to\infty$.
\end{theorem}
\begin{remark}
Let us note that the case $\alpha = 0$ of the above theorem corresponds to the codes of distance $\Theta(N/\log N)$ and dimension $\Theta(\log N)$ from Theorem~\ref{th:main1}.
\end{remark}

\subsection*{Lifted Product Codes}

In this paper, we continue our~study of the codes from~\cite{Panteleev&Kalachev:2019} in a more general form, and call them \emph{lifted product (LP) codes}. Roughly speaking, LP~codes are the lifted versions of hypergraph product codes proposed in~\cite{Tillich&Zemor:2009,Tillich&Zemor:2014}. 
As we will see later in Section~\ref{sc:lp}, LP codes generalize many well-known examples of  QLDPC codes~\cite{Hagiwara:2007, Tillich&Zemor:2009, Haah:2011, Kovalev&Pryadko:HBP:2013}, of which they are mostly motivated. We should also note that quasi-cyclic LP codes can be shown to be equivalent to a~special case of hyperbicycle codes~\cite{Kovalev&Pryadko:HBP:2013}, when the~parameter $\chi=1$ (see, more in Subsection~\ref{sc:qc-lp-codes}).

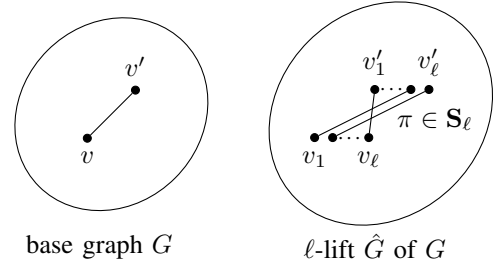
\begin{figure}
  \centering
  \begin{tikzpicture}[scale=0.8]
    \node[vertex] (v) at (-0.4,-0.4) [label=below:$v$]{};
    \node[vertex] (v') at (0.4,0.4) [label=$v'$]{} ;
    \draw (v) -- (v');
    \draw (0,0) ellipse[x radius=1.5 cm,y radius=1.7 cm, rotate=-45];    
    \node[empty]  (cap) at (-0.2,-2.2) {base graph $G$};
  \end{tikzpicture}
  \quad
  \begin{tikzpicture}[scale=0.8]
    \node[vertex] (v1)    at (-1.2,-0.4) [label=below:$v_1$]{};
    \node[vertex] (v2)    at (-0.9,-0.4) {};
    \node[empty]  (vdots) at (-0.6,-0.4) {\footnotesize\ldots};
    \node[vertex] (vell)    at (-0.3,-0.4) [label=below:$v_\ell$]{};

    \node[vertex] (v1')    at (-1.2+1,0.4) [label=above:$v'_1$]{};
    \node[empty]  (vdots') at (-0.9+1,0.4) {\footnotesize\ldots};
    \node[vertex] (v2')    at (-0.6+1,0.4) {};
    \node[vertex] (vell')    at (-0.3+1,0.4) [label=above:$v'_\ell$]{};
    \draw (v1) -- (v2');
    \draw (v2) -- (vell');
    \draw (vell) -- (v1');
    \draw (-0.1,0) ellipse[x radius=1.7 cm,y radius=2 cm, rotate=-45];
    \node[empty]  (pi) at (0.8,-0.1) {$\pi\in\mathbf{S}_\ell$};
    \node[empty]  (cap) at (-0.2,-2.2) {$\ell$-lift $\hat{G}$ of $G$};
  \end{tikzpicture}
  \caption{Lifting of the base graph $G$.}
  \label{fg:lifting}
\end{figure}

Large classical LDPC codes are often constructed as lifts of a~small graph called the~\emph{base graph} or the~\emph{protograph}~\cite{Thorpe:2003}. In graph theory, the~Tanner graphs~\cite{Tanner1981} of such $\ell$ times larger codes are called  \emph{$\ell$-lifts} or \emph{$\ell$-fold cover graphs} for the base graph. 
Let us remind that an~$\ell$-lift $\hat{G}$ of a~base graph\footnote{Multiple edges and loops are usually allowed in the base graph~$G$.} $G$ is obtained if we replace in the base graph each vertex $v\in V(G)$ with $\ell$~replicas $v_{1},\dots,v_{\ell}$; and replace each edge $e\in E(G)$ that connects vertices $v, v'\in V(G)$ with $\ell$~replicas~$e_1,\dots,e_\ell$ such that $e_i$ connects in $\hat{G}$ the vertices $v_i$ and $v'_{\pi(i)}$, where ${\pi\in\mathbf{S}_\ell}$ is some permutation on the set $\{1,\dots,\ell\}$ (see Fig.~\ref{fg:lifting}). 
Note that the~permutations for different edges may be different.
If~the~set of permutations $\pi$ is restricted to some finite permutation subgroup $\Gamma\subseteq \mathbf{S}_\ell$ such that\footnote{Such groups are obtained from some abstract finite group as the~group of all its left actions on itself.} $\abs{\Gamma}=\ell$, then we also say that $\hat{G}$ is a~\emph{$\Gamma$-lift} of $G$, and a~\emph{shift $\ell$-lift} when $\Gamma = \left<(1,2,\dots,\ell)\right>\subseteq S_\ell$ is the cyclic group of size $\ell$ generated by the permutation $(1,2,\dots,\ell)\in\mathbf{S}_\ell$.

Note that the parity-check matrix of an~LDPC code that was obtained as a~shift $\ell$-lift of some base graph is a quasi-cyclic matrix of circulant size~$\ell$. Let us briefly remind that quasi-cyclic (QC) matrices are block matrices, where each block is an~$\ell\times\ell$-circulant. They are usually represented by matrices over the quotient polynomial ring~$R_\ell=\F_2[x]/(x^\ell - 1)$. In~a~more general case of $\Gamma$-lifts the corresponding binary block matrices can be represented by matrices over a~group algebra~$\F_2 G$, where $G$~is an abstract group of order $\ell$ that is isomorphic to the permutation subgroup~$\Gamma\subseteq\mathbf{S}_\ell$.   

The idea of the~lifted product is to start from two small Tanner graphs $\cA$ and $\cB$ that have some shift $\ell$-lifts $\hat{\cA}$ and~$\hat{\cB}$, respectively. Let $A$ be the corresponding matrix over $R_\ell$ for~$\hat{\cA}$, and~$B$ be the corresponding matrix for $\hat{\cB}$. Since the~ring~$R_\ell$ is commutative, we will show in Section~\ref{sc:lp} that one can use a~slightly modified hypergraph product construction~\cite{Tillich&Zemor:2009} in order to obtain the~parity-check matrices $\HX$ and $\HZ$ over~$R_\ell$. Finally, we will see that $\HX$ and $\HZ$ (considered as binary block matrices) define a~CSS code denoted by $\LP(A,B)$. In~fact, the~idea of the lifted product is more general and can be used not only with the ring~$R_\ell$. Later we will show that this construction works for matrices $A$ and $B$ over any ring $R$ that is a~commutative $\ell$-dimensional $\F_2$-algebra. For example, if $G$ is an~abelian group of order~$\ell$, then the~group algebra $\F_2 G$ can be used as the ring~$R$. Hence, we can use not only shift $\ell$-lifts, but also $\Gamma$-lifts when the permutation group~$\Gamma$ is abelian.
In fact, the~general definition of lifted product codes, given in~Section~\ref{sc:lp}, can also be used with non-abelian groups as well. However, in this paper, we consider only abelian groups, while  the~non-abelian case is left for future work.

We should emphasize that the codes from~Theorem~\ref{th:main1} correspond to the case when $B$ is a~$1\times 1$ matrix with only one element $b\in R_\ell$. We denote the~code $\LP(A,B)$ by $\LP(A,b)$ in this case. Later we will show how to find or estimate the~dimension of $\LP(A,B)$ and $\LP(A,b)$ in many special cases.

In~Fig.~\ref{fg:main} you can see the parameters of the~LP~codes from Theorems~\ref{th:main1} and \ref{th:main2} (shown in red) against the~parameters of the~fiber bundle (FB) codes~\cite{Hastings:2021:fiber} (shown in green) and the~hypergraph product (HP) codes~\cite{Tillich&Zemor:2009} (shown in blue). In~fact, if we apply the~method used in Theorem~\ref{th:main2} to the fiber bundle codes, then we can also increase their dimension in the~same way as for LP codes. The~parameters of the quantum codes obtained in this way are also shown in green. We can see from Fig.~\ref{fg:main} that (up to polylog factors) the~parameters of the~all mentioned above codes converge to the~parameters of the~hypergraph product codes as the dimension $K$ grows asymptotically up to the code length~$N$.   

\begin{figure}
\centering
\begin{tikzpicture}[scale=4.2,>=stealth]
    \filldraw[red!10] (0,0.9) -- (0.1,0.9) -- (1,0.5) -- (0,0.5);
    \filldraw[green!10] (0,0.6) -- (0.6,0.6) -- (1,0.5) -- (0,0.5);
    \filldraw[blue!10] (0,0) -- (1,0) -- (1,0.5) -- (0,0.5);
    \draw[->, thick,red] (0.1,0.9)  -- (0.98,0.51);
    \draw[->, thick,green!60!black] (0.6,0.6)  -- (0.97,0.5);
    \draw[dashed,red] (0.,0.9) node[left,red]{$N/\log N$} -| (0.1,0) node[below,red] {$\log N$};
    \draw[dashed,blue] (0,0.5) node[left,blue]{$N^{1/2}$} -- (1,0.5) -- (1,0);
    \draw[dashed,green!60!black] (0,0.6) node[left,green!60!black]{$N^{3/5}/\polylog N$} -- (0.6,0.6) -- (0.6,0)node[below,green!60!black]{$N^{3/5}/\polylog N$};
    \draw (1,0) node[below,blue]{$N$} -- (1,0.02);
    \draw (1,0)|-(0,1);
    \filldraw[blue] (1,0.5) circle (0.02) node[right] {HP};
    \filldraw[red] (0.1,0.9) circle (0.02) node[above right] {LP};
    \filldraw[green!60!black] (0.6,0.6) circle (0.02) node[fill=white,above left=0.1,opacity=0,text opacity=1] {FB};
    \draw (0.8,0.8) node {\Huge ?};
     \draw[->] (0,0) -- (1.1,0) node[right] {$K$}; 
     \draw[->] (0,0) -- (0,1.1) node[above] {$d$};
\end{tikzpicture}
\caption{HP -- hypergraph product codes, FB -- fiber-bundle codes, LP -- lifted product code $\LP(A,1+x)$. The~parameters of all codes (the~minimum distance~$d$ and the~dimension~$K$) are shown in the~logarithmic scale up to polylogarithmic factors as the code length $N\to\infty$.}
\label{fg:main}
\end{figure}
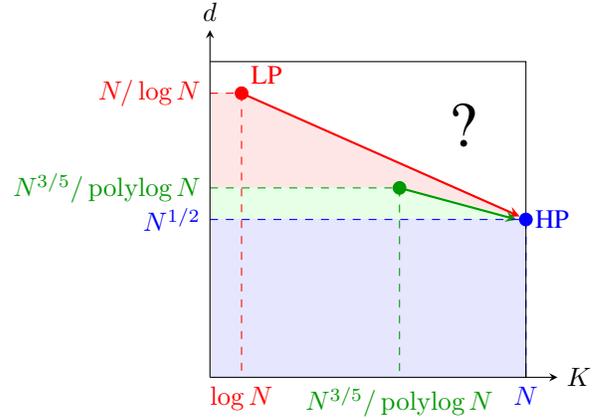

\subsection*{Lifted products of chain complexes}

Let us briefly show how to extend the~idea of the lifted product to chain complexes. It is known that $2$-dimensional chain complexes correspond to CSS codes~\cite{Bravyi:HMP:2014}. Nevertheless, \mbox{$s$-dimensional} chain complexes for $s>2$ can also be useful in the context of single-shot error correction~\cite{Earl:2019}. 

Consider some commutative ring~$R$. Let us remind that a~\emph{free $R$-module of rank $r$} is an~$R$-module~$M$, where there exists a~set of elements $\{m_1,\dots,m_r\}\subseteq M$ called \emph{basis} such that every $m\in M$ is uniquely represented as:
\[ m = a_1 m_1 + \dots + a_r m_r,\]
where $a_1,\dots,a_r\in R$.
Hence $M\cong R^r$, and if the ring $R$ is a~field, then $M$ is simply an~$r$-dimensional vector space over~$R$. A~canonical example of a~free $R$-module of rank $r$ is the~module of formal $R$-linear combinations of the~elements of some set~$S$, where $\abs{S}=r$.

By a~\emph{chain complex over a~commutative ring~$R$} we mean a~free $R$-module $\cC = \bigoplus_{i\in \Z} \cC_i$ with an~$R$-linear map $\partial\colon\cC \to\cC$ called a~\emph{boundary map} such that $\partial^2 = 0$, and $\partial(\cC_i) \subseteq \cC_{i-1}$. We suppose that each free $R$-module $\cC_i$ has finite rank, and $\cC_i = 0$ when $i < 0$ or $i > n$, where the parameter $n$ is called the~\emph{dimension} of $\cC$. We also assume that each $\cC_i$ comes with some preferred basis $\tilde{\cC}_i\subseteq \cC_i$, and we call its elements \emph{$i$-cells}. An~$n$-dimensional chain complex $\cC$ is usually written as
\[\cC_n \xrightarrow{\partial_{n}} \cC_{n-1} \xrightarrow{\partial_{n-1}} \cdots \xrightarrow{\partial_{2}} \cC_1 \xrightarrow{\partial_{1}} \cC_0,  \]
where $\partial_i = \partial|_{\cC_i}$, $i = 1,\dots,n$.

The lifted product of two chain complexes is obtained in a~similar way as for codes. We just consider the~standard tensor product $\cA \otimes \cB$ of chain complexes $\cA$, $\cB$ over a~commutative\footnote{For simplicity we define here the~lifted product only for commutative rings, though it may be easily extended to any ring $R$ if we consider a~tensor product $\cA\otimes\cB$ of a~free right $R$-module $\cA$ and a~free left $R$-module $\cB$.} ring~$R$ with boundary maps $\partial_{\cA}$, $\partial_{\cB}$, respectively. Thus if $R$ is in turn an~$\ell$-dimensional $\F_2$-algebra with some fixed basis\footnote{ For example, for $R_\ell$ the standard basis is $\tilde{R}_\ell = \{1,x,\dots,x^{\ell-1}\}$; for $\F_2G$ the~standard basis is $G$.} \[\tilde{R} = \{r_1,\dots,r_\ell\}\subseteq R,\] 
then the~boundary map~$\partial = \partial_{\cA}\otimes\id_\cB + \id_\cA\otimes\partial_{\cB}$ of $\cA \otimes \cB$ is an~\mbox{$R$-linear} map, and hence \mbox{an~$\F_2$-linear} map.  Therefore we can consider ${\cA \otimes \cB}$ also as a~chain complex over~$\F_2$, which we call the~\emph{lifted product} of the~chain complexes $\cA$ and $\cB$, and denote by $\cA \otimes_R \cB$ in order to emphasize the role of $R$. When $R=R_\ell$ we use a~shorter notation $\cA \otimes_\ell \cB$. \mbox{The~$n$-cells} of the~obtained chain complex $\cA \otimes_R \cB$ take the~form: $r a^{(i)} \otimes b^{(j)}$; where $r\in\tilde{R}$, $a^{(i)}$ is an~$i$-cell from $\tilde{\cA}_i$,  $b^{(j)}$ is a~$j$-cell from~$\tilde{\cB}_j$, and $i+j = n$. 

Note that any~matrix over an~$\ell$-dimensional $\F_2$-algebra~$R$ defines some binary linear code as its parity-check matrix\footnote{Any $m\times n$ matrix over $R$ can be also considered as an~$\ell m\times \ell n$ binary block matrix (see Section~\ref{sc:lp}).}. For example, any matrix over $R_\ell$ defines a~quasi-cyclic code. Any such linear code can be identified with the~corresponding \mbox{$1$-dimensional} chain complex $\cC_1\xrightarrow{\partial_{1}}  \cC_0$ such that $A$ is a matrix of the $R$-linear map $\partial_1$. 
Let $\cA$, $\cB$ be $1$-dimensional chain complexes over $R$ that correspond to the classical codes with parity-check matrices $A$, $B$ over $R$, respectively. Then it is not hard to see that the~CSS code~$\LP(A,B)$ defined in Section~\ref{sc:lp} corresponds to the~$2$-dimensional chain complex $\cA \otimes_R \cB$.



The remainder of the paper is structured as follows. \mbox{Section~\ref{sc:def}} contains some standard definitions and notations related to codes. In~Section~\ref{sc:lp} we give the definition of lifted product codes, where we also demonstrate that they contain many well-known QLDPC codes. Expanders are described in Section~\ref{sc:expanders}. Then we proceed with the proof of Theorem~\ref{th:main1} in Section~\ref{sc:lp-linear-dist}, and  in the~last section, we give some final remarks. The paper also contains three appendices, where we describe some well-known facts on the ring~$R_\ell$ (Appendix~\ref{sc:circulants}), study the decomposition of quasi-abelian LP~codes when the lift size is~odd (Appendix\footnote{Our main results do not rely on this supplementary material.}~\ref{sc:LP-decomp}), and give the~list of frequently used symbols and abbreviations (Appendix~\ref{sc:symbols}).

\section{Basic facts and definitions}\label{sc:def}

Here we fix notations and briefly recall some standard definitions related to classical and  quantum codes. More information can be found in a~survey~\cite{Babar:2015}.
In what follows, we assume that the reader is familiar with the standard algebraic objects like rings, fields, vector spaces, and modules (see~\cite{Dummit&Foote:2003} for a good reference). 

In this paper, it is convenient to consider vectors over a~field or a~ring as column vectors. Hence the~matrix-vector product is written as $Av$ instead of $A v^\T$. Besides, we denote by $\ker A$ and $\im A$ the~\emph{kernel} and the~\emph{image} of the~corresponding linear  operator $v\mapsto A v$, respectively. Note that $\im A$ coincides with the~\emph{column space} of the~matrix~$A$.
In~many places we use the standard notation $[n] = \{1,2,\dots,n\}$, where $n$ is a natural number. If $x,y$ are two binary vectors of length~$n$, then we denote by $x\cap y$ their~\emph{intersection}, i.e.,  the~vector~$x\cap y = (x_1y_1,\dots, x_n y_n)$. We say that an event $A_n$ occurs \emph{with high probability} (w.h.p) if $\mathbf{P}(A_n) \to 1$ as $n\to\infty$.  
Please, refer to Appendix~\ref{sc:symbols} for the list of symbols and abbreviations frequently used in our work. 

\subsection{Classical codes}

Consider a finite field\footnote{In this paper we consider only finite fields of characteristic $2$, but most of the results are valid for arbitrary finite fields.} $\F_q$ and an~\mbox{$n$-dimensional} vector space~$\F_q^n$ over~$\F_q$. A~\emph{linear $[n,k]_q$ code} is a~$k$-dimensional subspace $\cC\subseteq \F_q^n$, where the parameters~$n$ and $k$ are called the~\emph{length} and the~\emph{dimension} of $\cC$, respectively. We denote the~dimension $k$ of the~code~$\cC$ by $\dim \cC$. The~\emph{rate} of the~code~$\cC$ is equal to $k/n$.
The~elements of $\cC$ are called \emph{codewords}. \mbox{The~\emph{Hamming distance} $d(v, v')$} between vectors $v,v'\in \F_q^n$ is the number of positions in which they differ. The parameter 
\[d(\cC) = \min \{d(c, c') \mid c\ne c'; \ c, c'\in \cC\}\] 
is called the \emph{minimal distance} of $\cC$. By definition, we put~$d(\cC)=\infty$ when $k=0$. It is easy to see that $d(\cC)$ is equal to the minimal weight $\abs{c}$ of non-zero codewords, where the~\emph{weight~$\abs{c}$} is the number of non-zero components in~$c$. When $d(\cC)=d$ for a linear $[n,k]_q$ code $\cC$, we say that $\cC$ is an~$[n,k,d]_q$ code. 

A~linear $[n, k]_q$ code is usually defined either as the row space of a matrix $G$ called the \emph{generator matrix} or as the~kernel of a matrix $H$ called the \emph{parity-check matrix}. It~is easy to see that $GH^\T=\zm$,  $\rk G = k$, and $\rk H = n-k$. The~code defined by a~parity-check matrix $H$ is denoted by~$\cC(H)$. 

The vector space~$\F_2^n$ usually comes with the~standard scalar product $\langle x, y\rangle = x_1 y_1 + \dots + x_n y_n$. The \emph{dual code}~$\cC^\perp$ for a~linear $[n,k]_q$ code $\cC$ is the $[n,n-k]_q$ code 
\[
    \cC^\perp = \{ x \in \F_q^n \mid  \langle x, y \rangle = 0 \text{ for all }y\in\cC \}.
\]
It is not hard to see that a generator matrix for $\cC$ is a parity-matrix for $\cC^\perp$ and vice versa. 

Let $\pi\in \mathbf{S}_n$ be a permutation on the set~$[n]$. Given a~vector $v = (v_1,\dots, v_n)$, we denote by $\pi(v)$ the permuted vector $(v_{\pi(1)},\dots,v_{\pi(n)})$. We also extend this notation to sets of vectors of length~$n$ in a straightforward way: 
\[\pi(S) = \{\pi(v) \mid v\in S \}.\]
We say that two codes $\cC,\cC'\subseteq\F_q^n$ are (\emph{permutation}) \emph{equivalent} and write ${\cC \sim \cC'}$ if $\cC' = \pi(\cC)$ for some $\pi\in\mathbf{S}_n$. It is clear that equivalent codes have the same parameters $[n,k,d]_q$.

In~this paper we mostly deal with \emph{binary linear codes}, i.e., when $q=2$. In such cases we omit $q$ and simply write $[n,k]$ or $[n,k,d]$ code.

\subsection{Quantum CSS codes}\label{sc:stab-codes}
Consider the~$2^n$-dimensional Hilbert space $\C^{2^n}$, where the~$2^n$ standard basis vectors are indexed by binary vectors $u\in\F_2^n$ and denoted by $\ket{u}$. The space $\C^{2^n} = (\C^2)^{\otimes n}$ is usually called the \emph{$n$-qubit space}, where each component in the tensor product corresponds to one \emph{qubit}. 

A \emph{quantum code} $\cQ$ of \emph{length~$n$} and \emph{dimension~$k$} is \mbox{a~$2^k$-dimensional} subspace of~$\C^{2^n}$.  As in the classical case, we denote the dimension $k$ of the quantum code by $\dim\cQ$. In~\cite{CSS:1996, CSS2:1996} a~very important subclass of quantum codes called the~\emph{Calderbank-Shor-Steane (CSS) codes}, which is related to classical linear codes, was introduced. A~quantum \emph{CSS $[[n, k]]$ code}~$\cQ$ of length $n$ and dimension $k$ is defined by two classical linear codes $\CZ, \CX\subseteq \F_2^n$ such that  $\CX^\perp \subseteq \CZ$ and $k = \dim \CZ / \CX^\perp$ in the following way:  
\[\cQ = \spn_{\C} \bigl\{\textstyle \sum_{x \in \CX^\perp} \ket{z + x} \mid z \in \CZ \bigr\}. \]

It is easy to see that the property $\CX^\perp \subseteq \CZ$ is equivalent to $\CZ^\perp \subseteq \CX$. Moreover, if~$\HX$ is a parity-check matrix of~$\CZ$, $\HZ$ is the parity-check matrix of~$\CX$; then this property can be expressed as the following \emph{orthogonality condition}:

\begin{equation}\label{eq:comm-cond-CSS}
\HX \HZ^\T = \zm.
\end{equation}
Hence in order to define a~CSS code, we need two parity-check matrices $\HX$ and $\HZ$ such that \emph{every row of $\HX$ is orthogonal to every row of $\HZ$}. 
The dimension $k = \dim \cQ$ of the obtained quantum code~$\cQ$ is given by 
\begin{equation}\label{eq:CSS-dim}
    k = n - \rk \HX - \rk \HZ,
\end{equation}
since $k = \dim \CZ / \CX^\perp$. 

Given a~quantum~CSS code $\cQ$, we call the~codewords from the codes ${\CZ=\CZ(Q)}$ and~${\CX=\CX(Q)}$ the~\mbox{\emph{Z-codewords}} and \mbox{\emph{X-codewords}} of~$\cQ$, respectively.  Furthermore, the~\mbox{Z-codewords} from $\CX^\perp$ and the X-codewords from $\CZ^\perp$ are called \emph{degenerate}. This name can be explained if we interpret the~codewords from $\CZ$ and $\CX$ as undetected errors in a quantum system protected by the quantum code~$\cQ$. It can be shown that the~degenerate errors are precisely the~ones that don't change the state of the system. Therefore it makes sense to consider the~quotient spaces $\CZ/\CX^\perp$, $\CX/\CZ^\perp$ instead of $\CZ$, $\CX$. We say that codewords~$c,c'$ from the same equivalence class in these quotient spaces are \emph{equivalent} and denote this fact by $c \sim c'$. It is obvious that a~codeword $c$ is degenerate iff $c\sim \zm$, where $\zm$~is the zero vector. 
Let us note that the~degenerate codewords correspond to the~\emph{stabilizers} of the~quantum code~$\cQ$; while the~quotient spaces $\CZ/\CX^\perp$, $\CX/\CZ^\perp$ correspond to the~\emph{logical operators}, acting on the~quantum states protected by~$\cQ$.

Let us note that the spaces of degenerate Z-codewords $\CX^\perp$ and degenerate X-codewords $\CZ^\perp$ are generated by the rows of the parity-check matrices $\HX$ and $\HZ$, respectively. Hence the difference $c - c'$ of two equivalent codewords is always a linear combination of the rows from the corresponding parity-check matrix.

Since CSS codes have two types of codewords, they also have two types of minimum distances:
\begin{equation*}
    \dZ(\cQ) = \min_{z\in \CZ\setminus\CX^\perp} \abs{z}, \quad 
    \dX(\cQ) = \min_{x\in \CX\setminus\CZ^\perp} \abs{x}.
\end{equation*}

The minimum of these distances 
\[
    d(\cQ) = \min\{\dZ(\cQ),\dX(\cQ)\}
\] 
is called the \emph{minimum distance} of $\cQ$. As in the case of the~classical codes, we say that a~quantum CSS~$[[n, k]]$ code is an~$[[n,k,d]]$ code if $d(\cQ)=d$.

As in the case of classical codes, we also say that two CSS codes $\cQ,\cQ'$ of length $n$ are (\emph{permutation}) \emph{equivalent} and write $\cQ\sim \cQ'$ if $\CZ(\cQ) = \pi(\CZ(\cQ'))$ and $\CX(\cQ) = \pi(\CX(\cQ'))$ for some $\pi\in\mathbf{S}_n$.
It is clear that equivalent codes have the same parameters $[[n,k,d]]$, and, moreover, we see that ${\dZ(\cQ) = \dZ(\cQ')}$,  ${\dX(\cQ) = \dX(\cQ')}$. For any CSS code $\cQ$ we can also define the CSS code $\cQ^*$ with $\CZ(\cQ^*) := \CX(\cQ)$ and $\CX(\cQ^*) := \CZ(\cQ)$. It is obvious that:
\begin{equation}\label{eq:dual-CSS}
    \dZ(\cQ^*) = \dX(\cQ), \quad \dX(\cQ^*) = \dZ(\cQ).
\end{equation}

The CSS codes defined so far are binary quantum codes.
In~some cases, we also need \emph{non-binary CSS codes}.
They are defined by two matrices $\HX$ and $\HZ$ over $\F_q$ that satisfy equation~(\ref{eq:comm-cond-CSS}).  The definitions of dimension, minimum distance, degenerate codewords, equivalent codewords are obtained from the corresponding definitions for binary CSS codes if we replace $\F_2$ by $\F_q$. 



\subsection{Classical and quantum LDPC codes}

A classical \emph{low density parity check (LDPC) code}~\cite{gallager1963} is a linear code defined by a sparse binary parity-check matrix $H=\left(h_{i j}\right)_{m\times n}$. The sparseness usually means that the weights of all rows and columns in $H$ are upper bounded by some constant $w$ as the code length $n$ grows to infinity. It is helpful to define the~bipartite graph $\cT=\cT(H)$ called the \emph{Tanner graph}~\cite{Tanner1981}. In this graph the first part of nodes $v_1,\dots,v_n$ (called the \emph{v-nodes}) corresponds to the columns of $H$ (the variables), the second part of nodes $c_1,\dots,c_m$ (called the \mbox{\emph{c-nodes}}) corresponds to the rows of $H$ (the checks), and we connect a v-node $v_j$ with with a c-note $c_i$ whenever $h_{i j} = 1$, $i\in[m]$, $j\in [n]$. 

If the parity-check matrix $H$ is \emph{$(w_c, w_r)$-regular} (i.e., each column has weight $w_c$ and each row has weight $w_r$) then the corresponding Tanner graph is  also \emph{$(w_c, w_r)$-regular} (i.e., each v-node has degree $w_c$ and each c-node has degree $w_r$). We say that an LDPC code is \emph{$w$-limited} if the degree of each node in its Tanner graph is upper bounded by $w$. It is obvious that any LDPC code with $(w_c, w_r)$-regular parity-check matrix is $\max(w_c, w_r)$-limited.

In this paper by a \emph{quantum LDPC (QLDPC)} we mean a CSS $[[n, k, d]]$ code with sparse parity-check matrices $\HX$ and $\HZ$.
We can also introduce the~Tanner graph $\cT=\cT(\HX,\HZ)$ for any CSS $[[n, k, d]]$ code~$\cQ$ defined by $\HX$ and $\HZ$. In this case, the v-nodes correspond to $n$~qubits and the c-nodes to the rows of $\HX$ and $\HZ$ called the \emph{X-checks} and \emph{Z-checks}, respectively. We connect a v-node with a c-node if the corresponding qubit participates in the corresponding check.
Similar to the classic case we say that a QLDPC code is \emph{$w$-limited} if the degree of each node in its Tanner graph is upper bounded by $w$. This property is much more important in the quantum case due to the faulty nature of the current quantum hardware. It is clear that any CSS code with $(w_c, w_r)$-regular matrices $\mat{H}_X$ and $\mat{H}_Z$ is $\max(2w_c, w_r)$-limited.

\section{Lifted product}\label{sc:lp}

In this section, we give our~main definition of lifted product codes and show how it generalizes several known types of quantum codes~\cite{Tillich&Zemor:2009,Kovalev&Pryadko:HBP:2013,Haah:2011,Hagiwara:2007}. 
We also show how to estimate the~dimension of our codes in some important special cases, and give several examples. Finally, in~Subsection~\ref{sc:GHP}, we consider a~more specific case of our codes, first defined in~\cite{Panteleev&Kalachev:2019}, and show how to find the~dimension. Later, in Section~\ref{sc:lp-linear-dist}, we will use this special case to construct quantum LDPC with almost linear distance.

Before we move to the~description of our codes, we need some standard definitions and notations from algebra.  
Let $R$ be a ring. We denote the set of all $m\times n$ matrices over~$R$ by $\Mat_{m\times n}(R)$ or by $\Mat_{n}(R)$ in the case $m=n$. 
Consider a~field $\F$. In what follows by an~\emph{$\F$-algebra} we always mean \mbox{an~associative} algebra with a~multiplicative identity. It is well known~\cite[Theorem~1.3.1]{Drozd:1994} that every such algebra has a~faithful representation by $\ell\times\ell$ matrices over~$\F$, and for any element~$a\in R$ we denote by $\B(a)$  the~corresponding  $\ell\times\ell$ matrix over~$\F$. In the cases when $R$ is already a matrix ring over~$\F$ we assume that $\B(a) = a$.
Moreover, if $A=(a_{i j})_{m\times n}$ is a matrix over~$R$, then we consider the~corresponding block matrix \[ 
\B(A) = [\B(a_{i j})]_{m\times n}\in\Mat_{\ell m\times \ell n}(\F).
\]
It is easy to see that for any matrices $A$, $B$ over $R$ we have:
\begin{equation}\label{eq:block}
    \begin{split}
        \B(A B) &= \B(A)\B(B);
    \end{split}
\end{equation}

In this work we are mostly interested in the case when $R$ is a group algebra $\F_2G$ for some finite group~$G$. The~elements of $\F_2G$ are formal sums $\sum_{g\in G} \alpha_g g$, where $\alpha_g\in\F_2$. Consider elements $a=\sum_{g\in G}\alpha_g g$ and $b=\sum_{g\in G}\beta_g g$ from $\F_2 G$. Their sum $a+b$ and product $ab$ are defined as follows:
\[a+b = \sum_{g\in G} (\alpha_g + \beta_g) g, \quad ab = \sum_{g\in G}\bigg(\sum_{\substack{hr=g\\h,r\in G}}\alpha_h\beta_r\bigg) g.\]
If we index the~elements of the~group~$G=\{g_1,\dots,g_\ell\}$, then for every element $a=\sum_{g\in G}\alpha_g g\in  \F_2 G$  we define 
\begin{align*}
\bv(a) &= (\alpha_{g_1},\dots,\alpha_{g_\ell})\in \F_2^\ell;    \\
\B(a) &= \sum_{g\in G} \alpha_g\B(g), 
\end{align*}
where $\B(g)$ is the~permutation $\ell\times\ell$~matrix, defined as follows:
\[
\B(g)_{i j} = \begin{cases}
    1, & \text{ if } g_i = g g_j ;\\
    0, & \text{ otherwise}.
\end{cases}
\]
For every vector $v\in (\F_2G)^n$ we also consider the~block vector 
$\bv(v) = [\bv(v_1),\dots,\bv(v_n)]\in\F_2^{\ell n}$. 

It is not hard to see that for any $a\in \F_2 G$ the~weight of each row and each column of the~binary matrix~$\B(a)$ is the~same and equal to $\abs{\bv(a)}$.
Thus the~row and column weights of the~block matrix $\B(A)$, where $A=(a_{i j})_{m\times n}$ is a matrix over $\F_2 G$, can be easily found from the~corresponding \emph{weight matrix} (also called the~\emph{base matrix}) $W = W(A) = (w_{i j})_{m\times n}$, where $w_{i j} = \abs{\bv(a_{i j})}$. For example, $\B(A)$ is $w$-limited iff the~sum of elements of any row and column in $W$ is bounded above by $w$.
The matrix $W$ can be interpreted as the adjacency matrix for the~base Tanner graph $\cT$ that was used to obtain the~$G$-lifted Tanner graph $\hat{\cT}$ for the code $\cC(A)$ with the~parity-check matrix $\B(A)$, where $w_{i j}$ is equal to the number of edges between nodes $v_i$ and $v_j$ in the base Tanner graph~$\cT$.

Sometimes, where it does not cause confusion, we identify matrices and vectors over $R=\F_2 G$ with the~corresponding block matrices $\B(\cdot)$ and vectors $\bv(\cdot)$ over $\F_2$. For example, if we say that $A$ is $w$-limited, then it means that $\B(A)$ is \mbox{$w$-limited}. For~any vector $v\in R^n$ we denote by $\abs{v}$ the~Hamming weight $\abs{\bv(v)}$ of the corresponding block vector $\bv(v)\in\F_2^n$. We also often implicitly use the~following trivial equality:
\begin{equation}\label{eq:block-vec}
    \begin{split}
        \bv(Av) &= \B(A)\bv(v);
    \end{split}
\end{equation}
where $v$ is a vector, and $A$ is a matrix over~$\F_2G$. 

If $H$ is a~matrix over $R=\F_2 G$ where $\abs{G}=\ell$, then by $\cC(H)$ we denote the set
\[
\cC(H) = \{c\in R^n \mid Hc = \zm\}.    
\]
It is clear that the set~$\cC(H)$ is also a vector space over $\F_2$, and from~(\ref{eq:block-vec}) we see that it corresponds to the binary linear code $\cC(\B(H))\subseteq \F_2^{\ell n}$ defined by the binary block matrix $\B(H)$.

Let us note that if $G$ is abelian then $\F_2G$ is a~commutative ring.
Specifically, if $G$ is a~\emph{cyclic group~$\mathbf{C}_\ell$} of order $\ell$ generated by $x$, then $\F_2\mathbf{C}_\ell$ is isomorphic as a ring to the polynomial quotient ring $\F_2[x]/(x^\ell - 1)$. 
We~denote this ring by $R_\ell$ and usually represent its elements by polynomials in $x$.  If we index the~elements of the group~$\mathbf{C}_\ell$ as $g_i = x^{i-1}$, $i\in [\ell]$, then the~set of binary matrices $\B(a)$ where $a\in R_\ell$ is the~ring of circulant $\ell\times\ell$ matrices over $\F_2$. More details on $R_\ell$ can be found in Appendix~\ref{sc:circulants}.   

In this paper, with some~small abuse of terminology, \mbox{a~matrix} $A$ over $R_\ell$ and the~corresponding binary block matrix $\B(A)$ are called \emph{quasi-cyclic (QC)} of \mbox{\emph{lift size~$\ell$}} (also called the~\emph{circulant size}). Thus every~$\ell m\times \ell n$ binary QC~matrix of lift size~$\ell$ can be represented by some polynomial matrix $A\in\Mat_{m\times n}(R_\ell)$. The~class of QC~matrices is well known in coding theory since they are the parity-check matrices of \emph{quasi-cyclic} codes. At~the~same time, if $G$ is a~finite abelian group, then \mbox{matrices} over $\F_2 G$  and the corresponding binary classical error-correcting codes are called \emph{quasi-abelian}. Note that  most of the~best practical classical LDPC codes have QC~parity-check matrices. 
\begin{example}
Consider a~matrix $A\in\Mat_{2\times 3}(R_3)$ defined as follows:
\[ A = 
\left(\begin{array}{rrr}
1 & 0 & 1 + x^{2} \\
1 + x & 1 + x + x^{2} & x^{2}
\end{array}\right).
\]
It has the corresponding block matrix of lift size~$\ell=3$
\[
\B(A) = 
\left(\begin{array}{rrr|rrr|rrr}
1 & 0 & 0 & 0 & 0 & 0 & 1 & 1 & 0 \\
0 & 1 & 0 & 0 & 0 & 0 & 0 & 1 & 1 \\
0 & 0 & 1 & 0 & 0 & 0 & 1 & 0 & 1 \\
\hline
 1 & 0 & 1 & 1 & 1 & 1 & 0 & 1 & 0 \\
1 & 1 & 0 & 1 & 1 & 1 & 0 & 0 & 1 \\
0 & 1 & 1 & 1 & 1 & 1 & 1 & 0 & 0
\end{array}\right),
\]
and the~corresponding integer weight matrix
\[ W(A) = 
\left(\begin{array}{rrr}
1 & 0 & 2 \\
2 & 3 & 1
\end{array}\right).
\]
\end{example}

\subsection{Generalized bicycle (GB) codes}

The orthogonality condition~(\ref{eq:comm-cond-CSS}) from the~definition of CSS codes for the parity-check \mbox{matrices} $\HX$ and $\HZ$ is a~serious obstacle to designing good QLDPC codes using random-like constructions similar to the constructions used for classical LDPC codes. Thus it makes sense to consider large families of matrices of some particular form, where the orthogonality condition is always satisfied. One such quite general form for CSS codes was proposed in~\cite{Kovalev&Pryadko:HBP:2013}. We call these codes the~\emph{generalized bicycle~(GB)} codes since they include bicycle QLDPC codes~\cite{Mackay:2004} as a~special case.  Let~us briefly remind this construction. Consider two commuting binary $\ell\times \ell$ matrices $A$ and $B$, i.e., $AB = BA$. Let us define the~parity-check matrices as follows:
    \begin{equation}\label{eq:comm-ansatz}
    \HX = [A, B] \text{ and } \HZ = [B^\T, A^\T]. 
    \end{equation}

Then we see that $\HX \HZ^\T = AB + BA = \zm$. Hence the~commutativity condition (\ref{eq:comm-cond-CSS}) is always satisfied, and we obtain a~CSS code.  
It was proposed in~\cite{Kovalev&Pryadko:HBP:2013} to use binary circulant matrices $A$ and $B$ since they always commute. The corresponding class of codes includes the~bicycle codes from~\cite{Mackay:2004} as a special case when $B=A^\T$. 

Furthermore, we can obtain a more general class of codes if $A$ and $B$ are some $\ell\times\ell$ matrices representing elements from a~group algebra $\F_2G$ for an~abelian group $G$, ${\abs{G}=\ell}$. For example, the quasi-cyclic CSS codes from~\cite{Hagiwara:2007} can be considered as GB codes with $G = \mathbf{C}_P\times \mathbf{C}_{L/2}$. At~the same time, Haah's cubic codes~\cite{Haah:2011} can also be considered as GB codes with  $G=\mathbf{C}_L\times\mathbf{C}_L\times\mathbf{C}_L$, where $L$~is the lattice size.

\subsection{Hypergraph product (HP) codes}

Before we formally describe the LP codes in the next section, let us first remind the definition of \mbox{the~hypergraph} product (HP) codes from~\cite{Tillich&Zemor:2009}. Note that originally these codes were defined in terms of \mbox{hypergraphs}, but here it will be more convenient for us to give their definition in~a matrix form. 
\mbox{Suppose} that we have an~$[n_A, k_A, d_A]$ linear code~$\cC(A)$ and an~$[n_B, k_B, d_B]$ linear code~$\cC(B)$ with parity-check \mbox{matrices}\footnote{The parity-check matrices are not necessary full rank.} ${A\in\Mat_{m_A\times n_A}(\F_2)}$ and ${B\in\Mat_{m_B\times n_B}(\F_2)}$, \mbox{respectively}. Then the \emph{hypergraph product (HP)} code is the~CSS $[[N, K, d]]$ code denoted~$\HP(A,B)$ with the~parity-check matrices:
\begin{equation}\label{eq:hp-ansatz}
\begin{split}
    \HX &= [A\otimes I_{m_B}, I_{m_A}\!\otimes B],\\  
    \HZ &= [I_{n_A}\!\otimes B^\T, A^\T\otimes I_{n_B}],
\end{split}
\end{equation}
where the length~$N$ and the dimension~$K$ are as follows: 
\begin{equation}\label{eq:HP-dim}
\begin{split}
    N &= n_A m_B + n_B m_A,\\
    K &= 2k_A k_B - k_A(n_B - m_B) - k_B(n_A - m_A).
\end{split}
\end{equation}
As it was shown in~\cite{Tillich&Zemor:2009}, the minimum distance~$d$ of the hypergraph product code~$\HP(a,b)$ satisfies the following lower bound:
\[
d \ge \min(d_A, d_B, d_A^\T, d_B^\T),
\]
where the~parameters $d_A^\T$ and $d_B^\T$ are the minimal distances of the~\mbox{``transposed''} codes $\cC(A^\T)$ and $\cC(B^\T)$ defined by the~parity-check matrices $A^\T$ and $B^\T$, respectively. 

It is important to note that if the matrices $A$ and $B$ are $(w_c,w_r)$-limited then the parity-check matrices~$\HX$ and $\HZ$ of the code~$\HP(A,B)$ are \mbox{$w$-limited}, where $w = 2\max(w_c, w_r)$. Hence, using known asymptotically good families of classical LDPC codes with $(w_c,w_r)$-limited parity check-matrices, it is possible~\cite{Tillich&Zemor:2009} to construct $w$-limited CSS codes with asymptotically non-zero rate and $d=\Theta(\sqrt{N})$ as the code length~$N\to\infty$. 

\subsection{Non-binary HP codes}
Though HP~codes in the previous section are defined as binary CSS codes, it is also quite straightforward to define their non-binary versions over an~arbitrary finite field $\F_q$. Suppose that the characteristic of~$\F_q$ is $2$; then the~parity-check matrices $\HX$ and $\HZ$ for the~\emph{non-binary HP code} $\HP(A,B)$ are obtained from matrices $A$ and $B$ over $\F_q$ by (\ref{eq:hp-ansatz}) as for the case of binary HP codes. 

If the characteristics of $\F_q$ is not $2$, we need a slightly modified version of (\ref{eq:hp-ansatz}) in order to satisfy the orthogonality condition $\HX \HZ^\T$:

\begin{equation}\label{eq:hp-ansatz-general}
\begin{split}
    \HX &= [A\otimes I_{m_B}, -I_{m_A}\!\otimes B],\\  
    \HZ &= [I_{n_A}\!\otimes B^\T, A^\T\otimes I_{n_B}].
\end{split}
\end{equation}

\subsection{Lifted product (LP) codes}\label{sc:LP-code}

Here we consider a large family of quantum CSS codes that simultaneously generalize the GB codes and the HP codes.  These codes first appeared in our previous work~\cite{Panteleev&Kalachev:2019} in a~more restricted form under the~name \emph{generalized hypergraph product (GHP)} codes. In this work we present them in a~more general form and propose a~more informative name~---\emph{lifted product (LP) codes}. In some sense, we can view these codes as lifted versions of HP codes from~\cite{Tillich&Zemor:2009}, where we lift the~coefficients in the~matrices from the binary field~$\F_2$ up to some ring $R$ that is also a finite dimensional $\F_2$-algebra. Let~us remind that the elements of $R$ can be represented by binary $\ell\times\ell$ matrices. Hence when we define the~LP code over $R$ we identify $R$ with the~corresponding matrix ring\footnote{In each example of a~finite dimensional $\F_2$-algebra below we always provide the~corresponding matrix representation.}. Therefore without loss of generality in the definition below we assume that $R\subseteq \Mat_\ell(\F_2)$, and $\B(a) = a$ for every $a\in R$.

Thus LP codes are essentially HP codes, where we replace elements in their binary matrices $A$ and $B$ by some elements of a~matrix ring~$R\subseteq \Mat_\ell(\F_2)$. As the result, we have matrices $A\in \Mat_{m_A\times n_A}(R)$ and $B\in \Mat_{m_B\times n_B}(R)$ over some matrix ring~$R\subseteq \Mat_\ell(\F_2)$. If $M = (m_{ij})_{m\times n}$ is a matrix over $R$ we can consider its \emph{conjugate transpose} $M^* = (m^\T_{ji})_{n\times m}$, where $m^\T_{ji}$ is the standard transpose of the matrix $m_{ji}\in \Mat_\ell(\F_2)$. Let us emphasize that $\B(M^*) = \B(M)^\T$.
Now, as in the case of HP codes, we also introduce matrices: 
\begin{equation}\label{eq:lp-ansatz}
\begin{split}
    \HX &= [A\otimes I_{m_B}, I_{m_A}\!\otimes B],\\    
    \HZ &= [I_{n_A}\!\otimes B^*, A^*\otimes I_{n_B}].
\end{split}
\end{equation}

These matrices have coefficients from the matrix ring $R$, but we can consider the~corresponding binary block matrices $\B(\HX)$ and $\B(\HZ)$. In order to define a~CSS code, we need to make sure that these block matrices satisfy the~orthogonality condition  $\B(\HX)\B(\HZ)^\T = \zm$. Since $\B(\HZ)^\T = \B(\HZ^*)$, using (\ref{eq:block}) we see that condition is equivalent to ${\HX \HZ^* = \zm}$, and thus can be rewritten as:
\begin{align*}
    [A\otimes I_{m_B}, I_{m_A}\!\otimes B]
    \begin{bmatrix}
    I_{n_A}\!\otimes B\\
    A\otimes I_{n_B} \\
    \end{bmatrix}
    &= \zm,
\end{align*}
which can be further simplified to 
\begin{align}\label{eq:temp-cond}
    (A\otimes I_{m_B}\!)(I_{n_A}\!\otimes B) + (I_{m_A}\!\otimes B)(A\otimes I_{n_B}\!) = \zm.
\end{align}
It is not hard to see that if \emph{every element of $A$ commutes with every element of $B$} (we call such matrices \emph{element-wise commuting}) then using the~mixed-product formula 
\[(X\otimes Y)(X'\otimes Y') = (XX'\otimes YY')\] 
for the Kronecker product~$\otimes$ we have:
\begin{align*}
    (A\otimes I_{m_B}\!)(I_{n_A}\!\otimes B) = (I_{m_A}\!\otimes B)(A\otimes I_{n_B}\!) = A \otimes B.
\end{align*}
Since $R$ is a ring of characteristic\footnote{If $\ch R \ne 2$ we should define $\HX = [A\otimes I_{m_B}, -I_{m_A}\!\otimes B]$.} $2$, condition~(\ref{eq:temp-cond}) is satisfied, and
every pair of element-wise commuting matrices $A$ and $B$ defines the CSS code with the~parity-check matrices $\B(\HX)$ and $\B(\HZ)$. We denote this code by $\LP(A,B)$ and call the~\emph{lifted product (LP) code}. In what follows $\HX$ and $\HZ$ are also called the~\emph{parity-check matrices} for $\LP(A,B)$

One can easily verify that if we take $R=\F_2\cong \Mat_1(\F_2)$ then LP~codes coincide with HP~codes. We can also see that the generalized bicycle (GB) codes with two commuting $\ell\times \ell$ matrices $A$ and $B$ given in~(\ref{eq:comm-ansatz}) are also a special case of LP codes if we consider the matrices $A$ and $B$ as $1\times 1$ matrices over~$\Mat_\ell(\F_2)$.

The code length $N$ of the CSS code $\LP(A, B)$ is given by $N = \ell(n_A m_B + n_B m_A)$. We are not aware of any simple way to find the dimension~$K$ of the code~$\LP(A, B)$ in the general case, but it is possible to find or estimate $K$ in some special cases. For example, if $m_A < n_A$ and $m_B > n_B$, then by counting the number of rows in $\B(\HX)$ and $\B(\HZ)$ we obtain from~(\ref{eq:CSS-dim}) the~following lower bound:
\begin{equation*}
    K  \ge \ell(n_A - m_A)(m_B - n_B).
\end{equation*}
Below we consider some other examples.
\begin{example}\label{ex:HP-non-bin}
    If $R$ is a finite field $\F_q$ with $q=2^r$ elements and $A$, $B$ are some matrices over $\F_q$, then the~code $\LP(A,B)$ is obtained from the non-binary code $\HP(A,B)$ if we replace each non-binary element $\alpha$ in the parity-check matrix $\HX$ of $\HP(A,B)$ by the corresponding associated $r\times r$ binary matrix\footnote{The finite field $\F_q$, $q=2^r$, is an $r$-dimensional vector space over $\F_2$.  The~associated matrix $M_\alpha$ for $\alpha\in\F_q$ is the matrix of the $\F_2$-linear transform $x\mapsto \alpha x$.}~$M_\alpha$. At the~same time, in the~parity-check matrix $\HZ$ each non-binary element $\alpha$ is replaced by $M_\alpha^\T$. Thus we obtain binary block matrices that define the binary CSS code $\LP(A,B)$. It is clear that the length of the binary code $\LP(A,B)$ is $r$ times the length of the corresponding non-binary code~$\HP(A,B)$. It is also not hard to check that the~dimension of $\LP(A,B)$ is also $r$ times bigger:
    \begin{equation}\label{eq:LP-non-bin-dim}
        \dim \LP(A, B) = r\dim\HP(A,B).
    \end{equation}
\end{example}

\begin{example}\label{ex:LP}
If $A$ is some $m\times n$ matrix over a~commutative ring $R\subseteq \Mat_\ell(\F_2)$ and $B = A^*$, then we have the~LP code $\LP(A) = \LP(A, A^*)$ of length 
\begin{equation*}
N=\ell(n^2 + m^2)     
\end{equation*}
and dimension 
\begin{equation*}
K\ge \ell (n - m)^2,     
\end{equation*}
with parity-check matrices
\begin{equation*}\label{eq:lp1-ansatz}
\begin{split}
    \HX &= [A\otimes I_{n}, I_{m}\otimes A^*],\\    
    \HZ &= [I_{n}\otimes A, A^*\otimes I_{m}].
\end{split}
\end{equation*}
The lower bound for~$K$ easily follows from the~CSS dimension formula~(\ref{eq:CSS-dim}) if we take into account that matrices $\HX$ and $\HZ$ have no more than $2\ell m n$ rows in total. Moreover, if $A$ is full rank (as a binary block matrix), then the~matrices $A\otimes I_{m_B}$ and  $I_{n_A}\otimes A$ are also full rank. Hence all rows in $\B(\HX)$ and $\B(\HZ)$ are independent and we see that:
\begin{equation*}
K = \ell (n - m)^2.     
\end{equation*}
\end{example}

\subsection{Quasi-cyclic and quasi-abelian LP codes}\label{sc:qc-lp-codes}
One simple way to make all matrices over $R$ to be element-wise commuting is to enforce $R$ to be a~commutative ring.  In~this section we consider one particularly important special case of LP codes with commutative ring $R=R_\ell$ that we call \emph{quasi-cyclic (QC) LP codes}, and its generalization called \emph{quasi-abelian (QA) LP codes} when $R = \F_2 G$ for some finite abelian group~$G$.
As we saw at the beginning of this section, if $R$ is one of these rings we can easily control the density of the parity-check matrices $\B(\HX)$ and $\B(\HZ)$ by looking at the weight matrices $W(\HX)$ and $W(\HZ)$.  
In what follows, we are going to identify matrices and vectors over $R=\F_2 G$ with the~corresponding block matrices $\B(\cdot)$ and vectors $\bv(\cdot)$ over $\F_2$.

One big advantage of QC LP codes is that they are \mbox{constructed} from classical QC~LDPC codes. There are many examples of such codes with very good parameters.
\begin{example}
Consider the~$[155,64,20]$ QC LDPC code $\cC(A)$ of circulant size $\ell=31$ from~\cite{Tanner:2001}.
\[
A = \left(\begin{array}{rrrrr}
x & x^{2} & x^{4} & x^{8} & x^{16} \\
x^{5} & x^{10} & x^{20} & x^{9} & x^{18} \\
x^{25} & x^{19} & x^{7} & x^{14} & x^{28}
\end{array}\right)
\]
We can construct the~$8$-limited code~$\LP(A) = \LP(A,A^*)$ (see~Example~\ref{ex:LP}) with parameters $[[1054, 140, d]]$. We should note that after an~extensive simulation under the~BP-OSD decoder~\cite{Panteleev&Kalachev:2019} we have not found any non-degenerate codeword of weight less than $20$.
\end{example}

If an~element~$a = \sum_{g\in G} \alpha_g g\in \F_2 G$ is represented by the~matrix~$\B(a)\in\Mat_\ell(\F_2)$, then we have $\B(a)^\T = \B(\bar{a})$,  where $\bar{a} = \sum_{g\in G} \alpha_g g^{-1}\in \F_2 G$. The~map $a\mapsto \bar{a}$ is called the~\emph{antipode map} for $\F_2G$. If the~group~$G$ is abelian, then the~antipode map is an~automorphism of $\F_2G$ that respects the~weight of elements, i.e., for any $u,v\in\F_2G$ we have:
\[
\overline{u+v} = \bar{u}+\bar{v},\quad 
\overline{uv} = \bar{u}\bar{v},\quad
\abs{\bar{u}} = \abs{u}.
\]
For example, for any $a = a_0 + a_1 x+\dots +a_{\ell-1}x^{\ell-1}\in R_\ell$ we~have $\bar{a} = a_{0} + a_{\ell-1}x +\dots +a_{1}x^{\ell-1}$, i.e., $\bar{a} = x^\ell a(x^{-1})$.
If $A=(a_{i j})_{m\times n}$ is a matrix over $\F_2G$, then it is clear that the~conjugate transpose defined in subsection~\ref{sc:LP-code} is the matrix $A^*=(\bar{a}_{j i})_{n\times m}$. Since the antipode map is an~automorphism of $\F_2 G$, one can easily check that:
\begin{equation}\label{eq:rk-A}
    \rk_{\F_2} A = \rk_{\F_2} A^\T = \rk_{\F_2} A^* = \rk_{\F_2} \bar{A},
\end{equation}
where $\bar{A} = (\bar{a}_{i j})_{m\times n} = (A^*)^\T$. At the~same time, since
the~map $u\mapsto \bar{u}$ just permutes the bits in $u$, it is not hard to verify that $\LP(A,B) \sim \LP(\bar{A},\bar{B})$. Besides, for $\cQ = \LP(A,B)$ we see from equation~(\ref{eq:lp-ansatz}) that when we change the~roles of~the matrices $\HX$ and $\HZ$, we obtain the~code $\cQ^*$ that is permutation equivalent to $\LP(B^*,A^*)$ and also to $\LP(A^*,B^*)$. Hence we have that: 
\[
\cQ^* \sim \LP(A^*,B^*) \sim \LP(\overline{A^*},\overline{B^*}) = \LP(A^\T,B^\T),
\]
and by (\ref{eq:dual-CSS}) we have:
\begin{equation}\label{eq:LP-dual-dim}
\begin{split}
    \dZ(\LP(A,B)) &= \dX(\LP(A^\T, B^\T));\\        
    \dX(\LP(A,B)) &= \dZ(\LP(A^\T, B^\T)).        
\end{split}
\end{equation}
Equations~(\ref{eq:LP-dual-dim}) allow us to reduce a~problem of finding $\dX(\LP(A,B))$ to a~similar problem for $\dZ(\LP(A,B)$ and vice versa.

Let us note that QC LP codes are permutation equivalent to a~special case of hyperbicycle codes~\cite{Kovalev&Pryadko:HBP:2013}. If~we let ${\chi=1}$ in equation (19) from~\cite{Kovalev&Pryadko:HBP:2013}, then the~obtained CSS code is permutation equivalent to the~code $\LP(A,B)$; where $\ell:=c$, $A:=\sum_i a_i x^i$, $B:=\sum_i b_i x^i$. We should also note that some general results on hyperbicycle codes, such as Theorem~3 from~\cite{Kovalev&Pryadko:HBP:2013}, can be applied to QC LP codes if we assert that $\chi=1$.

\subsection{Special case of QC LP codes}\label{sc:GHP}

Let us describe a more specific case of QC LP codes used in~\cite{Panteleev&Kalachev:2019} in order to construct  several examples\footnote{In~\cite{Panteleev&Kalachev:2019} these codes were called GHP codes. The parity-check matrices of three examples were described in the appendix.} with good error-correcting performance in the~depolarizing channel. In~\cite{Panteleev&Kalachev:2019} for simplicity we considered only the case when the lift size is odd.  Here we consider the general case.

Let polynomial $b\in\F_2[x]$ be an~irreducible factor of $x^\ell - 1$ and $A=(a_{i j})_{m\times n}$ be a matrix over $R_\ell$. Consider the code~$\LP(A, b)$, where we understand $b$ as a~$1\times 1$ matrix over $R_\ell$. Hence the~parity-check matrices for this code have the~following form\footnote{See also equation~(\ref{eq:qc-ghp}) in the introduction.}
\begin{equation}\label{eq:lp-simple}
\HX = [A, b I_m], \quad \HZ = [\bar{b} I_n, A^*].    
\end{equation}

We~denote by $\phi_b$ the~homomorphism  from $R_\ell$ to the~quotient ring~$R_{(b)} = \F_2[x] / (b)$ given by the map~$u \mapsto u \mod b$. Since $b$ is an irreducible polynomial over $\F_2$, the quotient ring~$R_{(b)}$ is isomorphic to the~finite field~$\F_q$, where $q=2^{\deg b}$. Thus we can describe this homomorphism as 
\[
\phi_b(u) = u(\beta) = u_0 + u_1\beta + \dots + u_{\ell-1}\beta^{\ell-1}\in \F_q,
\] 
where $\beta\in \F_q$ is a root of $b$ and ${u = \sum_{i=0}^{\ell-1} u_i x^i\in R_\ell}$. For~example, if~$b = 1 + x$ then the~ring~$R_{(b)} = \F_2[x]/(1+x)$ can be identified with $\F_2$ and $\phi_b(u) = u(1) = u_0 + \dots +u_{\ell - 1}$ is just the number of ones modulo $2$ in the binary vector~$u$. 
The~homomorphism $\phi_b$ can be also extended to vectors and matrices over $R_\ell$ if we apply it to each element. \mbox{For~a~vector $v$} and a~matrix~$M$, the~result of its action is denoted by $v(\beta)$ and $M(\beta)$, respectively.


\begin{lemma}\label{lp-simple-dim}
Let $b\in\F_2[x]$ be an irreducible factor of $x^\ell - 1$, and $A$ be~a~matrix over $R_\ell$. The dimension of the~code~$\LP(A,b)$ is~equal to
\[\dim \LP(A,b) = \deg b\rbr{\dim \cC(A(\beta)) +  \dim \cC(A^\T(\beta))},\] 
where $\beta$ is a~root of the~polynomial~$b$ in the field~$\F_q \cong R_{(b)}$. 
\end{lemma}
\begin{proof}
Trivially, using (\ref{eq:rk-A}) we have: 
\begin{align*}
\dim_{\F_2}\!\ker H^\T_{\mathrm{X}} &= \ell m - \rk_{\F_2}\!\HX;\\  
\dim_{\F_2}\!\ker H^*_{\mathrm{Z}} &= \ell n - \rk_{\F_2}\!\HZ.  
\end{align*}
Besides, from~(\ref{eq:lp-simple}), taking into account that ${H^\T_{\mathrm{X}} = \left[\begin{smallmatrix} A^\T\!\\ \!bI_m\! \end{smallmatrix}\right]}$, and ${H^*_{\mathrm{Z}} = \left[\begin{smallmatrix} bI_n\\ A \end{smallmatrix}\right]}$,  it follows that:
\begin{align*}
    H^\T_\mathrm{X} u = \zm  &\iff A^\T u = \zm, bu = \zm; \\
    H^*_{\mathrm{Z}} u = \zm &\iff A u = \zm, bu = \zm;
\end{align*}
and we have\footnote{Here $\cC(H) = \{c \mid Hc = 0\}$ is considered as a vector space over $\F_2$.} $\ker H^\T_{\mathrm{X}} = \cC(A^\T)\cap\cC_b^m$, $\ker H^*_{\mathrm{Z}} = \cC(A)\cap\cC_b^n$, where $\cC_b = \{c\in R_\ell \mid b c = 0\}$ is the~cyclic $[\ell, \deg b]$-code defined by the~parity polynomial~$b$. Clearly, ${g = (x^\ell - 1)/b}$ is the~generator polynomial for $\cC_b$, and the~elements of the~finite field~$R_{(b)} \cong \F_q$ are in one-to-one correspondence, given by the~map~${r\mapsto gr}$, with the~elements of $\cC_b$. Indeed, since ${gb = 0}$, and hence $g(r + b R_\ell) = gr$, we see that this one-to-one correspondence is defined correctly. Moreover, for every $r\in R_\ell$ it follows that $g\phi_b(r) = g(r + bR_\ell) = gr$, where we used that the~homomorphism $\phi_b\colon R_\ell\to\F_2[x] / (b)$ can be also defined as $\phi_b\colon r\mapsto r + b R_\ell$. Therefore for every $c=gu\in \cC_b^n$ we have:
\[
c\mkern-2mu\in\mkern-2mu \cC(\!A\mkern-1mu) \Longleftrightarrow  gAu = 0  \Longleftrightarrow g\phi_b(\!A u\mkern-1mu) = 0  \Longleftrightarrow u\mkern-2mu \in\mkern-2mu \cC(\phi_b(\!A)),
\]
and the~map~$u\mapsto gu$ also gives a one-to-one correspondence between $\cC(A)\cap\cC_b^n$ and $\cC(\phi_b(A))$. Since ${R_{(b)}\cong \F_q}$,  ${q=2^{\deg b}}$,  we see that $\dim_{\F_2}\! \cC(A)\cap\cC_b^n = \deg b \cdot \dim \cC(A(\beta))$. By exactly the same arguments as before, we also find that $\dim_{\F_2}\! \cC(A^\T)\cap\cC_b^m = \deg b \cdot \dim \cC(A^\T(\beta))$.
Finally, using CSS dimension formula~(\ref{eq:CSS-dim}) we get:
\begin{align*}
\dim \LP(A,b) &= \ell(n + m) - \rk_{\F_2}\!\HX - \rk_{\F_2}\!\HZ\\
&= \dim_{\F_2}\!\ker H^\T_\mathrm{X} + \dim_{\F_2}\!\ker H^*_{\mathrm{Z}}\\
&= \dim_{\F_2}\! \cC(A^\T)\cap\cC_b^m + \dim_{\F_2}\! \cC(A)\cap\cC_b^n\\
&= \deg b\rbr{\dim \cC(A^\T(\beta)) + \dim \cC(A(\beta))},
\end{align*}
and the lemma is proved.
\end{proof}

\begin{remark}
An~alternative proof of Lemma~\ref{lp-simple-dim} in the case of odd~$\ell$ can be found in Appendix~\ref{sc:LP-decomp}. 
\end{remark}

\begin{remark}
If $b=1+x$ then $R_{b}\cong \F_2$, $\beta = 1$, and we obtain a~slightly simpler dimension formula: 
\begin{equation}\label{eq:lp-simple-dim}
   \dim \LP(A,1+x) = \dim \cC(A(1)) +  \dim \cC(A^\T(1)). 
\end{equation}
We should also emphasize that in this case the~cyclic code~$\cC_b$ is the~$[\ell,1,\ell]$~repetition code with $g=\sum_{i=0}^{\ell-1} x^i$.
\end{remark}

\section{Expanders}\label{sc:expanders}

In this section, we remind some standard definitions and facts related to expander graphs and codes. A~more detailed treatment of these subjects can be found in~\cite{Hoory:2006}. The~main theoretical result of this section is Proposition~\ref{pr:alpha-beta}, which gives us a~way to 
construct quasi-cyclic matrices of very larger lift sizes with good expansion properties. We use this construction to obtain our main result in Section~\ref{sc:lp-linear-dist}. As a~byproduct of this construction, we also get Corollary~\ref{cor:QC-dist}, which shows us how to get asymptotically good families of classical quasi-cyclic LDPC codes with a~close to optimal lift size.

\subsection{Expander graphs}

Let $G$ be a~simple\footnote{Simple graphs do not have loops and multiple edges.} graph with the set of vertices $V(G)$ and the set of edges $E(G)$. If vertices $v, v'\in V(G)$ are connected by an edge $e\in E(G)$, we call $v, v'$ \emph{adjacent} and denote this fact by $v\sim v'$ or by $v\sim_e v'$ when we want to emphasize the edge~$e$. We also say that a~vertex~$v\in V(G)$ is \emph{incident} to an~edge~$e\in E(G)$ if $v$ is one of the two vertices that~$e$ connects. The \emph{degree} of a vertex $v$ denoted by~$\deg v$ is the number of edges connected to it. A graph~$G$ is called \emph{$w$-regular} if all its vertices have degree~$w$.  The \emph{adjacency matrix} for a~graph~$G$ with $V(G)=\{v_1,\dots,v_n\}$ is the matrix $A(G) = (a_{i j})_{n\times n}$, where $a_{i j}$ is the number of edges $e\in E(G)$ such that $v_i \sim_e v_j$.
Since $A(G)$ is a symmetric matrix, it has $n$ real-valued eigenvalues $\lambda_1 \ge \dots \ge \lambda_n$. Let~$\lambda(G) = \max(\abs{\lambda_2}, \abs{\lambda_n})$. We call an~$n$-vertex~$w$-regular graph $G$ an~\emph{$(n, w, \lambda)$-expander} if $\lambda(G) = \lambda$.

For any $S\subseteq V(G)$ we denote by $E(S)$ the set of internal edges for $S$, i.e., 
\[
    E(S) = \{e\in E(G)  \mid \exists v,v'\in S\colon  v \sim_e v' \}.
\]
In the next section, we will need the following very well-known property of the expander graphs.
\begin{lemma}\label{lm:EML}
    If $G$ is an~$(n, w, \lambda)$-expander and $\abs{S} \le \alpha n$, then 
    \[ \abs{E(S)} \le \frac12\rbr{\alpha + \frac{\lambda}{w}} w\abs{S}.\]    
\end{lemma}
\begin{proof}
If $G$ is an~$(n, w, \alpha)$-expander and $S\subseteq V(G)$, then by the~expander mixing lemma~\cite[Lemma~2.5]{Hoory:2006} we have: 
\begin{equation}\label{eq:EML}
        \abs{\abs{E(S, S)} - \frac{w \abs{S}\abs{S}}{n}} \le \lambda \sqrt{\abs{S}\abs{S}},
\end{equation}
where $E(S,S) = \{ (v,v')\in S\times S \mid  v \sim v'\}$. Let us emphasize that each edge $e\in E(G)$ that connects $v,v'\in S$ gives two different pairs $(v,v')$ and $(v',v)$ in the~set~$E(S,S)$. Hence $\abs{E(S,S)} = 2\abs{E(S)}$ and we obtain:
\[
    2\abs{E(S)} \le \frac{w}{n} \abs{S}^2 + \lambda\abs{S} \le \rbr{\frac{\abs{S}}{n} + \frac{\lambda}{w}}w\abs{S}.
\]
Since $\abs{S} \le \alpha n$, we have $ \abs{E(S)} \le \frac12\rbr{\alpha + \frac{\lambda}{w}} w\abs{S}$.
\end{proof}

\subsection{Expanding binary matrices}

We say that a binary $m\times n$ matrix $H$ is \emph{$(\alpha, \beta)$-expanding}, where $\alpha,\beta$ are some positive real numbers, if for all $x\in \F_2^n$  such that $\abs{x}\le \alpha n$  we have $\abs{H x} \ge \beta \abs{x}$. It~is obvious that if $H$ is a parity-check matrix of some code~$\cC$, then $d(\cC) > \alpha n$. Furthermore, we also say that an $m\times n$ matrix $A$ over $R_\ell$ is \emph{\mbox{$(\alpha,\beta)$-expanding}} if the~corresponding binary block \mbox{matrix} $H=\B(A)$ is \mbox{$(\alpha,\beta)$-expanding}. It is clear that if $\B(A)$ is a~parity-check matrix of some QC code~$\cC$, then $d(\cC) > \alpha \ell n$. 

The following important proposition shows that for a~wide range of parameters there exists a~$w$-limited QC matrix $A$ such that $A$ and $A^\T$ are both $(\alpha,\beta)$-expanding.

\begin{proposition}\label{pr:alpha-beta}
    For every $\eps \in (0,1)$ there exist $\alpha$, $\beta$, $\gamma$, $w$ such that for any natural numbers $\ell > 1$ and $n \ge \gamma\ln \ell$ there exists a~\mbox{$w$-limited} QC~matrix $A\in \Mat_{m\times w n}(R_\ell)$, $m\le \eps w n$, such that the~matrices~$A$ and $A^\T$ are both $(\alpha, \beta)$-expanding.
\end{proposition}

In order to prove this proposition, we need to describe some particular type of Tanner codes~\cite{Tanner1981} used in~\cite{Sipser&Spielman:1996, Zemor:2001}. 
Consider a simple $w$-regular graph $G$ with $2n$ vertices and a linear ${[w, w-r]}$~code~$\cC_0$. The~idea of the~Tanner code $\cT(G, \cC_0)$ is to assign its code bits to the~$wn$ edges of $G$ and for each vertex~$v\in V(G)$ constrain the~bits connected to $v$ by the~code~$\cC_0$. More formally, if we index the~edges of $G$ by the~set~$[wn]$ and for each vertex~$v\in V(G)$ denote by~$N(v)$ the~set of indexes for the~edges connected to~$v$; then the~Tanner code is defined as
\[
\cT(G,\cC_0) = \{c\in \F_2^{w n} \mid \forall v\in V(G)\colon c|_{N(v)}\in\cC_0 \},
\]
where $c|_{N(v)}$ is obtained from $c$ by deleting all the~bits with indexes outside of~$N(v)$. 

We suppose that $\cC_0$ always comes with the some fixed parity-check matrix\footnote{Since $\cC_0$ is a~$[w,w-r]$ code, the~rows of $H_0$ are linearly independent.} $H_0\in \Mat_{r \times w}(\F_2)$.
The parity-check matrix~$H$ for the~code~$\cT(G,\cC_0)$ consists of $2n$ groups of rows $(\mathcal{R}_v)_{v\in V(G)}$, where each group $\mathcal{R}_v$ corresponds to the~$r$ parity-checks of $\cC_0$ related to the vertex $v\in V(G)$, i.e., $\rho|_{N(v)}$ is one of the rows from~$H_0$ for~$\rho\in \mathcal{R}_v$.
Hence $H\in\Mat_{2rn\times wn}(\F_2)$, and the~code~$\cT(G, \cC_0)$ has non-zero dimension whenever $2r < w$. Moreover, it is not hard to see that if $2r < w$ then $H$ is a~$w$-limited matrix. 

\begin{remark}\label{rm:QC-tanner}
Let us warn the~reader that the~code~$\cT(G, \cC_0)$ depends on how we index  the~edges of $G$ by the~set~$[wn]$. Moreover, the~parity-check matrix $H$ for the~code $\cT(G, \cC_0)$ also depends on how we order its rows. 
Let $\hat{G}$ be an~$\ell$-lift of a~smaller base graph $G$, obtained by fixing in $G$ an~arbitrary edge orientation and assigning to each oriented edge $e=(v,v')$ some shift $s=s(e)$, which corresponds to the~permutation $\pi\in\mathbf{S}_\ell$ (the~right cyclic shift  by $s$ positions) from Fig.~\ref{fg:lifting}. 
Denote by $h(e)$ (resp. $h'(e)$) the~column of $H_0\in \Mat_{r \times w}(\F_2)$, corresponding to the~connection of the~edge $e$ to the~vertex $v$ (resp. $v'$) in the~Tanner code $\cT(G, \cC_0)$. The~parity-check matrix of this code looks as follows:
\[ H = \left(\begin{array}{ccc}
     & \mathbf{0} & \\
     \cdots & h'(e) & \cdots\\
     & \mathbf{0} & \\
     \cdots & h(e) & \cdots\\
     & \mathbf{0} & 
   \end{array}\right)
\]
if we arrange its~rows to make each group of rows  $(\mathcal{R}_v)_{v\in V(G)}$ a~contiguous block.
Then it is not hard to see that $\cT(\hat{G}, \cC_0)$ is a~QC~code of lift size~$\ell$ with the~following QC parity-check matrix $\hat{H}$, defined as a~matrix over the~ring~${R_\ell=\F_2[x]/(x^\ell - 1)}$: 
\[ \hat{H} = \left(\begin{array}{ccc}
     & \mathbf{0} & \\
     \cdots & x^{s(e)} h'(e) & \cdots\\
     & \mathbf{0} & \\
     \cdots & h(e) & \cdots\\
     & \mathbf{0} & 
   \end{array}\right). 
\]
\end{remark}

The~next lemma shows that if $G$ is a~sufficiently good expander and $\cC_0,\cC_0^\perp$ have relatively large minimum distances, then both $H$~and~$H^\T$ are $(\alpha,\beta)$-expanding for any size of~$G$.
\begin{lemma}\label{lm:alpha-beta-exp}
    Let $H$ be the parity-check matrix of~$\cT(G,\cC_0)$, where $G$ is a~$(2n, w, \lambda)$-expander\textup{;} $\cC_0$~is~a~${[w, w - r,d]}$~code.
    If~$\lambda < \delta w$, ${d \ge \delta w}$ and $d(\cC_0^\perp) \ge \delta w$, then the~binary matrices $H$, $H^\T$ are  \mbox{$(\alpha, \beta)$-expanding} for all $\alpha < \frac{\delta}{w}\rbr{1-\frac{\lambda}{\delta w}}$ and ${\beta=\frac{1}{\delta w}\rbr{\delta-\alpha w-\frac{\lambda}{w}}}$.
\end{lemma}
\begin{proof}
Let us start with a quick remark that the~conditions ${\lambda < \delta w}$ and ${\alpha < \frac{\delta}{w}\rbr{1-\frac{\lambda}{\delta w}}}$ imply that $\frac{\delta}{w}\rbr{1-\frac{\lambda}{\delta w}} > 0$ and $\beta > 0$.
We~divide the~proof into two parts. In~the~first part we show that the~matrix~${H\in\Mat_{2rn\times wn}(\F_2)}$ is \mbox{$(\alpha, \beta)$-expanding}, while in the second part that the same holds for~$H^\T$.
    
Consider a~binary vector $x\in\F_2^{wn}$ such that $\abs{x} \le \alpha w n$. Let~$X\subseteq E(G)$ be the corresponding set of edges, where the~bits from~$x$ are equal to~$1$. We divide the set of vertices incident to some edges from $X$ into two parts: the~set~$S$ of vertices incident to \emph{at least} $\delta w$ edges from $X$; and the~set~$S'$ of vertices incident to \emph{less than} $\delta w$ edges from~$X$.  One~can easily check\footnote{Each edge from $X$ is connected to at most $2$ vertices from $S$.} that $\abs{S} \le 2 \abs{X}/\delta w$. Hence from $\abs{X} \le \alpha wn$ it also follows that $\abs{S}\le 2\alpha n/\delta$, and we can estimate the~number of edges from $X$ connected only to $S$ by~Lemma~\ref{lm:EML}:
\begin{equation*}
    \abs{E(S)}  \le \frac12\rbr{\frac{\alpha}{\delta}+\frac{\lambda}{w}}w\abs{S}\le \rbr{\frac{\alpha}{\delta^2}+\frac{\lambda}{\delta w}}\abs{X}.        
\end{equation*}
Therefore we have:
\[
      \abs{X\setminus E(S)}\ge \abs{X}-\abs{E(S)}\ge \rbr{1-\frac{\alpha}{\delta^2}-\frac{\lambda}{\delta w}} \abs{X}.
\]
For each edge from $X\setminus E(S)$ one of the two vertices it connects is outside of $S$; hence this vertex is in~$S'$.  Since $S'$ is connected to less than $\delta w$ edges from $X$, we can estimate the size of $S'$ as follows:
\[
    \abs{S'} > \frac{\abs{X\setminus E(S)}}{\delta w} \ge \frac{1}{\delta w}\rbr{1-\frac{\alpha}{\delta^2}-\frac{\lambda}{\delta w}}\abs{X}.
\]
Since $d(\cC_0) \ge \delta w$, for each $v\in S'$ the parity-checks of the code $\cC_0$ that correspond to $v$ can not be simultaneously satisfied. Therefore we have  
\begin{multline*}
  \abs{Hx} \ge \abs{S'} > \frac{1}{\delta w}\rbr{1-\frac{\alpha}{\delta^2}-\frac{\lambda}{\delta w}}\abs{X} \\
  \ge \frac{1}{\delta w}\rbr{1-\frac{\alpha w}{\delta}-\frac{\lambda}{\delta w}}\abs{X}
  = \frac{\beta}{\delta} \abs{x} \ge \beta\abs{x},    
\end{multline*}
where we used that $1/\delta \le w$ and $\delta \le 1$.
Hence we see that $\abs{x} \le \alpha w n$ implies $\abs{H x} \ge \beta \abs{x}$. Thus $H$ is $(\alpha, \beta)$-expanding, and the first part of the proof is complete.

Now let us prove that $H^\T$ is also $(\alpha, \beta)$-expanding. Hence we need to show that
for any $y\in \F_2^{2rn}$ such that $\abs{y}\le 2\alpha rn$ we have $\abs{H^\T y} \ge \beta \abs{y}$. As we already mentioned,  $H$ consists of $2n$ groups of rows $(\mathcal{R}_v)_{v\in V(G)}$, where each group $\mathcal{R}_v$ corresponds to the~$r$ parity-checks of~$\cC_0$ related to the vertex $v\in V(G)$. Thus $H^\T y= \sum_{v\in S} \rho_v$, where $\rho_v$ is a linear combination of the~rows from the group $\mathcal{R}_v$ and $S$ is the set of vertices $v\in V(G)$ such that $\mathcal{R}_v$ contains at least one row from the linear combination $H^\T y$. Since $\rho_v|_{N(v)}$ is a~linear combination of rows from~$H_0$ and $d(\cC_0^\perp) \ge \delta w$, we see\footnote{We have $\rho_v \ne 0$ since the~rows in $H_0$ are linearly independent.} that $\abs{\rho_v} \ge \delta w$ for all $v\in S$. Let us note that $\abs{\rho_v\cap\rho_{v'}}=0$ for $v\ne v'$  unless the vertices~$v,v'$ are connected by an~edge from $E(S)$, in which case we have $\abs{\rho_v\cap\rho_{v'}} \le 1$. Moreover, $\abs{\bigcap_{v\in I}\rho_v} = 0$ if $\abs{I} > 2$.  Therefore we obtain that 
\begin{multline*}
\abs{H^\T y} = \Big|\sum_{v\in S} \rho_v\Big| = \sum_{v\in S} \abs{\rho_v} - 2\!\!\!\sum_{\substack{v\ne v'\\v,v'\in S}}\abs{\rho_v\cap\rho_{v'}} \\
\ge \sum_{v\in S} \abs{\rho_v} - 2\abs{E(S)}  
\ge \delta w\abs{S} - 2\abs{E(S)}.    
\end{multline*}
Since  $\abs{S}\le \abs{y} \le 2\alpha r n$, if we apply Lemma~\ref{lm:EML} to the~set~$S$, we obtain that $\abs{E(S)} \le \frac12\rbr{\alpha r + \frac{\lambda}{w}}w\abs{S}$ and
\begin{multline*}
\abs{H^\T y} \ge \delta w\abs{S} - 2\abs{E(S)} \ge \rbr{\delta - \alpha r - \frac{\lambda}{w}}w\abs{S}\\
\ge \rbr{\delta - \alpha w - \frac{\lambda}{w}}\abs{y} = \delta w\beta \abs{y} \ge
\beta \abs{y},    
\end{multline*}
where we used that $w \ge r$, $w\abs{S} \ge r\abs{S}\ge\abs{y}$, and $\delta w \ge 1$. Hence we proved that $\abs{y}\le 2\alpha rn$ implies $\abs{H^\T y} \ge \beta \abs{y}$ and the second part is complete. 
\end{proof}

\begin{proof}[Proof of Proposition~\ref{pr:alpha-beta}]

It is known~\cite{Friedman:2003, Puder:2015} that \mbox{a~random} \mbox{$w$-regular} graph $G$ w.h.p. has $\lambda(G) < 2\sqrt{w - 1} + 1$. Thus for any~sufficiently large $n\ge n_0$ there exists\footnote{We should also mention the reference~\cite[Theorem~1.3]{Alon:2021} where an~explicit construction of such graphs is given.} \mbox{a~$w$-regular} graph~$G$ with $2n$~vertices such that ${\lambda(G) < 2\sqrt{w - 1} + 1}$. At~the~same time, it is known~\cite{Agarwal:2019} that for some positive constants $c_1,c_2$ for a random shift \mbox{$\ell$-lift} $\hat{G}$ of $G$ we have $\lambda(\hat{G}) \le c_1\lambda(G)$ with probability at least $1 - \ell\exp(-c_2 n/w^2)$. Hence if we choose a~sufficiently large\footnote{It is enough to use $\gamma = \max(\ceil{w^2/c_2}, n_0)$.}~$\gamma$ this probability is positive for all $\ell > 1$, $n \ge \gamma \log_2 \ell$; and therefore there exists a~shift $\ell$-lift $\hat{G}$ of $G$ such that  
\begin{equation}\label{eq:lift-exp}
    \lambda(\hat{G}) < c_1\!\rbr{2\sqrt{w - 1} + 1}.
\end{equation}

Now consider $\eps \in (0, 1)$, and let $\delta$ be some real number  from the~interval $(0,1/2)$ such that the following inequality holds:
\begin{equation}\label{eq:GV}
\max(1- \eps/2, \eps/2) + \eps/4 \ <\ 1 - h_2(\delta),    
\end{equation}
where $h_2(x) = -x\log_2 x - (1-x)\log_2(1-x)$ is \mbox{the~binary} \mbox{entropy} function. Let $\cC_0$ be a~code defined by a~uniformly random parity-check matrix~$H_0\in\Mat_{r\times w}(\F_2)$, where ${r = \floor{\frac{\eps}{2} w}}$. Since $1-r/w\to 1 - \eps/2$ and $r/w\to  \eps/2$ as ${w\to\infty}$, taking into account (\ref{eq:GV}),  the~Gilbert–Varshamov bound implies that for \emph{every} sufficiently large $w$ w.h.p. we have 
\begin{equation}\label{eq:code-ine}
d(\cC_0) \ge \delta w,\quad d(\cC_0^\perp) \ge \delta w,    
\end{equation}
and the~matrix $H_0$ is full rank. Indeed, let us remind that the~Gilbert–Varshamov bound can be proved~\cite{Barg&Forney:2002} using a~code defined either by a~random parity-check matrix or a~random generator matrix (i.e., the~parity-check matrix of the~dual code), and w.h.p. those matrices are full rank. Therefore by the~union bound we see that w.h.p $\cC_0$ is a~$[w,w-r]$ code that satisfy (\ref{eq:code-ine}) as $w\to\infty$.
Hence if we consider a~Tanner code~$\cT(\hat{G},\cC_0)$ with the code $\cC_0$ that satisfy (\ref{eq:code-ine}), and choose  a sufficiently large $w$ such that we also have
\[c_1\!\rbr{2\sqrt{w - 1} + 1} < \delta w,\] 
then using (\ref{eq:lift-exp}) by~Lemma~\ref{lm:alpha-beta-exp} we obtain that the~matrices $H, H^\T$ are $(\alpha, \beta)$-expanding for some positive constants $\alpha,\beta$; where $H$ is the~parity-check matrix of the~code~$\cT(\hat{G},\cC_0)$.
Since $\hat{G}$ is a~shift $\ell$-lift for $G$, according to Remark~\ref{rm:QC-tanner} we can assume that the~matrix~$H$ is a~QC matrix of lift size~$\ell$, i.e., $H = \B(A)$ for some \mbox{$w$-limited} matrix~$A\in\Mat_{m\times wn}(R_\ell)$, where $n \ge \gamma \log_2 \ell$, and  $m = 2\floor{\frac{\eps}{2} w}n\le \eps w n$. Now since the~matrices $H=\B(A)$, $H^\T = \B(A^*)$ are $(\alpha,\beta)$-expanding, we see that $A$ and $A^*$ are $(\alpha,\beta)$-expanding. Finally, since the antipode map $u\mapsto \bar{u}$ is an~automorphism of $R_\ell$, and $|A^\T u|  = |\overline{A^\T u}| = |A^* \bar{u}|$, we also obtain that the matrix~$A^\T$ is $(\alpha,\beta)$-expanding, and the~proof is complete.
\end{proof}


\subsection{Asymptotically good QC LDPC codes with large lift sizes}
Proposition~\ref{pr:alpha-beta} can be used to construct asymptotically good families of classical QC~LDPC codes of very large lift sizes~$\ell$. Indeed, if we put $n = \ceil{\gamma \ln \ell}$ and $\eps = 1 - R$ in~Proposition~\ref{pr:alpha-beta}, then we obtain the code~$\cC(A)$ of rate at~least $R$ and distance at~least $\alpha N$, where $N$ is the~code length. Moreover since $\log N \sim \log \ell$, and $n = \Theta(\log N)$ as~$N\to\infty$, we obtain the~following corollary.  

\begin{corollary}\label{cor:QC-dist}
    For any $R < 1$ there exists a~family of classical QC~LDPC codes of length $N$ and rate at least $R$ with distance~$\Omega(N)$ and lift size~$\Omega(N/\log N)$ as~$N\to\infty$.
\end{corollary}

We will see below that the lift size $\Omega(N/\log N)$ from Corollary~\ref{cor:QC-dist} is in some sense the~best possible for QC LDPC codes and even for quasi-abelian LDPC codes with linear minimum distance. More specifically, we will show using the~results from~\cite{Smarandache:2012} that: for any~family of quasi-abelian LDPC~codes of distance $\Omega(N)$, defined by $m\times n$ parity-check matrices with $m < n$ over commutative group algebras, the~lift size grows at most like $O(N/\log N)$ as the~code length~$N\to\infty$. 

Before we prove this we need some definitions and notations from~\cite{Smarandache:2012}. If $A$ is an $m\times n$ matrix and $I\subseteq [n]$, then we denote by $A_{I}$ the $m\times\abs{I}$ submatrix of~$A$ that contains only  the columns of $A$ with indexes from the~set~$I$. If $X$ is a~finite set and $f$ is a real-valued function on $X$, we denote by $\underset{x\in X}{\mathrm{min}^*} f(x)$ the \emph{minimum nonzero value} of $f(x)$ on~$X$; if there are no nonzero values then $\min^*$ gives $+\infty$. Let us remind that the~\emph{permanent} of an~integer matrix~$A = (a_{i j})_{n\times n}$ denoted by $\perm A$ is given by the following formula:
\[
    \perm A = \sum_{\pi\in \mathbf{S}_n} a_{1, \pi(1)} \dots a_{n, \pi(n)}. 
\]
Thus $\perm A$ is essentially $\det A$ if we ignore the signatures of the permutations $\pi\in \mathbf{S}_n$. If all elements from $A$ are non-negative integers then we have the following trivial upper bound:
\begin{equation}\label{eq:perm-triv}
    \perm A \le \prod_{j\in [n]} (a_{1 j} +\dots + a_{n j}).
\end{equation}

If we have a~matrix $H = (h_{i j})_{m\times n}$ over $\F_2G$, where $G$ is some abelian group of size $\ell$; then we can consider its~weight matrix $W = (w_{i j})_{m\times n}$, where $w_{i j} = \abs{h_{i j}}$ is the~row and the~column weight of each $\ell\times\ell$ block in $H$ (considered as a binary block matrix), $i\in [m], j\in [n]$. If we fix the~weight matrix $W$ and consider matrices $H$ with  $\ell\to\infty$, it is natural to expect that $d(\cC(H)) \to\infty$. However it turns out~\cite[Theorem~7]{Smarandache:2012} that when $m < n$ there is an upper bound\footnote{In~\cite{Smarandache:2012} the bound is proved only for QC~matrices, but the way it is proved in~\cite{Butler&Siegel:2013} can be easily extended to matrices over abelian group algebras.} on the minimum distance $d$ of the code~$\cC(H)$ that depends only the weight matrix $W$, which implies that $d$ doesn't grow with the lift size~$\ell$. The upper bound is as follows:
\begin{equation}\label{eq:QC-perm}
    d\le \underset{\substack{S\subseteq [n]\\ \abs{S}=m+1}}{\mathrm{min}^*} \sum_{i\in S} \perm W_{S\setminus i}.
\end{equation}
\begin{remark}
We should emphasize that if $m=n$, then this bound does not work anymore, and the~distance~$d$ can grow linearly with the~lift size~$\ell$. For example, if $H  \in\Mat_{n}(R_\ell)$ is given by the formula:
\[
\mat{H} =
\begin{pmatrix}
1 & 1 & 0 & \ldots & 0\\
0 & 1 & 1 & \ldots & 0\\
\hdotsfor{5}\\
0 & 0 & \ldots & 1 & 1\\
x & 0 & \ldots & 0 & 1\\
\end{pmatrix},
\]
then $H$ is a~$2$-limited matrix, but $d=\ell n$.
\end{remark}

Let us now return to the case when $m < n$.
If the matrix $H$ is $w$-limited and\footnote{The $w$-limited matrices with $w\le 1$ define trivial codes with minimum distance $1$.} $w\ge 2$, then the sum of elements in any row and column of the weight matrix~$W$ is bounded above by $w$. Hence the same holds for every its submatrix $W_{S\setminus i}$ from bound (\ref{eq:QC-perm}). Thus using~(\ref{eq:perm-triv}) we obtain from~(\ref{eq:QC-perm}) that $d \le (m+1)w^m$. Now suppose that the~minimum distance $d$ of the~code~$\cC(H)$ is~$\Omega(N)$ as $N\to\infty$, i.e., $d \ge \alpha N$ for some fixed $\alpha > 0$. In~this case we have: 
\[
\alpha N \le (m+1)w^m \le 2^m w^m\le w^{2m} < w^{2n}.
\]
Hence $\alpha N < w^{2n}$, $n = \Omega(\log N)$, and finally we obtain that $\ell = N/n$ is bounded above by~$O(N/\log N)$ as $N\to\infty$.

\section{LP codes with almost linear distance}\label{sc:lp-linear-dist}

This section contains the~proof of our main result. We consider the~special case of QC LP codes, studied in Section~\ref{sc:GHP}, and prove Proposition~\ref{pr:main}, which is the~main technical tool to establish the~lower bound on the~code distance. Finally, we combine this lower bound with our results on expanding matrices from Section~\ref{sc:expanders} to show the~existence of QLDPC codes with almost linear distance.    

\begin{proposition}\label{pr:main}
Let $A\in\Mat_{m\times n}(R_\ell)$ be a~$w$-limited QC matrix such that $A$ and $A^\T$ are $(\alpha, \beta)$-expanding. Consider a~quantum code~$\cQ = \LP(A,1+x)$. There exists a~constant $\gamma > 0$, that depends only on $\alpha$, $\beta$, and $w$, such that\textup{:}
\begin{enumerate}
    \item $d(\cQ) \ge \gamma\ell$\textup{;}
    \item If $\dim \cC(A^\T(1)) = 0$ then $\dZ(\cQ) \ge \gamma \ell n$\textup{;}
    \item If $\dim \cC(A(1)) = 0$ then $\dX(\cQ) \ge \gamma \ell m$.
\end{enumerate}
\end{proposition}

In order to prove Proposition~\ref{pr:main} we need two simple lemmas. Below by $\ones{\ell}$ we denote  the~\emph{all~one polynomial}~$\sum_{i=0}^{\ell-1} x^i$.
\begin{lemma}\label{lm:LP-simple-cws}
    Consider a quantum code~$\cQ = \LP(A, 1+x)$, where $A\!\in\!\Mat_{m\times n}(R_\ell)$, and let ${B=A(1)\!\in\!\Mat_{m\times n}(\F_2)}$. Then every non-degenerate codeword ${[u, v] \in\CZ(\cQ)\setminus\CX^\perp(\cQ)}$  satisfies one of the following two conditions\textup{:}
\begin{enumerate}
    \item  $u(1)$ is a non-zero codeword from $\cC(B)\subseteq \F_2^n$\textup{;}
    \item $[u,v]\sim [\zm, \ones{\ell}v']$ for some ${v'\in \F_2^m\setminus \im B}$\textup{;} and we have $u = (1+x)h$, $v = \ones{\ell}v' + Ah$, where ${h\in R_\ell^n}$, $\abs{h_i} \le \ell/2$, $i\in[n]$. 
\end{enumerate}    
\end{lemma}
\begin{proof}
First, we describe the~equivalence classes of codewords in~$\CZ=\CZ(Q)$. From the~definition of $\cQ$ it easily follows that 
\begin{equation}\label{eq:lp-simple-cws}
\begin{split}
[u, v]\in\CZ  &\iff Au = (1+x)v; \\   
[u, v]\in\CX^\perp  &\iff \exists h\in R_\ell^n\colon u = (1+x)h, v = A h;    
\end{split}
\end{equation}
and $[u, v]\in\CZ$ is equivalent to $[u',v']\in\CZ$ iff there exists $h\in R_\ell^n$ such that 
\begin{equation}\label{eq:lp-simple-degen}
     u' - u = (1+x)h, \quad v' - v = Ah.
\end{equation}
Hence if $[u, v]\in \CZ(\cQ)$ then $A(1)u(1) = \zm$, and we have $u(1)\in \cC(B)$. Therefore when $u(1)\ne \zm$ we obtain that $[u,v]$ satisfies the~first condition of the~lemma. 

Now let us suppose that $u(1) = \zm$. In this case we see that ${u = (1+x)h}$ for some $h\in R_\ell^n$. Let us note that we can always choose $h$ such that $\abs{h_i} \le \ell/2$, $i\in[n]$. Indeed, if it does not have this property, then we can replace it with $h'$ such that $(1+x)h' = (1+x)h$, defined by
\[
h'_i = 
\begin{cases}
h_i, & \text{if } \abs{h_i} \le \ell/2;\\
h_i + \ones{\ell}, & \text{if } \abs{h_i} > \ell/2. 
\end{cases}
\]
Hence we can assume that we have $h$ with the desired property, and from (\ref{eq:lp-simple-degen}) it follows that $[u,v] \sim [\zm, r]$, where $r = v + Ah$. Since $[\zm, r]$ is a non-degenerate codeword from~$\CZ(\cQ)$, we see that $r\ne \zm$, and $(1+x)r = \zm$. Thus we have $r = \ones{\ell}v'$, where $v'\in\F_2^m$, and obtain that:
\[
v = r + Ah = \ones{\ell}v' + Ah.
\]
We claim that $v'\not\in \im B$. Indeed, otherwise $v' = Bs$ for some $s\in\F_2^n$, and it is easy to see that  $r = A \ones{\ell} s$ in this case. Hence we have ${[\zm,r]\sim [\zm,\zm]}$, and obtain~a contradiction with the fact that $[\zm,r]$ is a~non-degenerate codeword. This proves that $v'\not\in \im B$, and $[u,v]$ satisfies the~second condition of the~lemma.
\end{proof}

\begin{lemma}\label{lm:autcorr}
    For any vector~$a\in R_{\ell}^n$, where $\abs{a_i} \le \ell/2$,  $i\in [n]$, there exists $t$ such that $\abs{(1 + x^t)a} \ge \abs{a}$.
\end{lemma}
\begin{proof}
Since $|a| = |x^t a|$ for any $t$, we have 
\[|(1+x^t) a|=|a|+|x^t a|-2|a\cap x^t a| = 2(|a|-|a\cap x^t a|).\]
It is not hard to see that:
\[
\sum_{i=0}^{\ell-1}|a\cap x^i a|=\sum_{j=1}^{n}\sum_{i=0}^{\ell-1}|a_j\cap x^i a_j| = \sum_{j=1}^n |a_j|^2.
\]
Since $|a_i| \le \ell/2$,  $i\in [n]$, we obtain
\[
\sum_{i=0}^{\ell-1}|a\cap x^i a|=\sum_{j=1}^n |a_j|^2\le \sum_{j=1}^n |a_j|\frac{\ell}{2}=\frac{1}{2}|a|\ell.
\]
Therefore we have:
\[\sum_{t=0}^{\ell-1}\!|(1+x^t)a|=2\!\sum_{i=0}^{\ell-1}\!|a|-2\!\sum_{t=0}^{\ell-1}\!|a\cap x^t a| \ge 2|a|\ell-|a|\ell=|a|\ell,\]
and there should exists $t$ such that $\abs{(1 + x^t)a} \ge \abs{a}$.
\end{proof}

\begin{proof}[Proof of Proposition~\ref{pr:main}]
Since $A$ and $A^\T$ are $(\alpha,\beta)$-expanding matrices, we see that 
\[
d(\cC(A)) \ge \alpha \ell n, \quad d(\cC(A^\T)) \ge \alpha \ell m.
\]
Let $B = A(1)\in \Mat_{m\times n}(\F_2)$ be the base matrix for the QC matrix~$A$. It is clear that if $c\in \cC(B)$ then $\ones{\ell} c\in \cC(A)$. Hence we obtain that $d(\cC(A)) \le \ell d(\cC(B))$ and by the same argument $d(\cC(A^\T)) \le \ell d(\cC(B^\T))$. Thus we have
\[
d(\cC(B)) \ge \alpha n, \quad d(\cC(B^\T)) \ge \alpha m.
\]
Consider a~non-degenerate codeword~$c=\![u, v]\!\in\! \CZ(\!\cQ)\!\setminus\!\CX^\perp\!(\!\cQ)$. From Lemma~\ref{lm:LP-simple-cws} it follows that $u(1)\in \cC(B)$, and we have only two cases: 
\begin{enumerate}
    \item $u(1) \ne \zm$, and thus $\abs{u(1)} \ge \alpha n$, since $u(1)\in \cC(B)$;
    \item $u(1) = \zm$, and thus $c\sim [0,\ones{\ell} v']$ for some $v'\in \F_2^m\setminus\{\zm\}$.
\end{enumerate} 

Let us consider each case separately.

\textbf{Case 1.} In this case we have $\abs{u(1)} \ge \alpha n$. Let us show that $\abs{c} = \abs{u} + \abs{v} \ge \gamma_1\ell n$, where $\gamma_1 = \min(\alpha/2, \alpha\beta/4)$.  If~$\abs{u} > \alpha \ell n/2$ then we are done. Now suppose we have $\abs{u} \le \alpha \ell n/2$. We claim that $\abs{v}\ge \alpha\beta \ell n/4$. Indeed, consider  
\[u^{(t)} = \ones{t}u, \quad s^{(t)} = A u^{(t)},\]
where $\ones{t} = \sum_{i=0}^{t-1}x^i$, $t\in [\ell]$. Since for any $t\in [\ell]$ we have $(1+x)\ones{t} = 1 + x^t$,   it follows that:
\[
s^{(t)} = A u\ones{t} = (1+x)v\ones{t} = (1+x^t)v, 
\]
where we use that $Au = (1+x)v$ by (\ref{eq:lp-simple-cws}).
Besides, we see that $\abs{u^{(1)}} = \abs{u} \le \alpha \ell n/2$ and $\abs{u^{(\ell)}} = \abs{\ones{\ell} \cdot u(1)} \ge \alpha \ell n$. Hence we can consider the minimal $t_0$ such that $\abs{u^{(t_0 + 1)}} \ge \alpha \ell n$. Since $u^{(t_0 + 1)} = u^{(t_0)} + x^{t_0 + 1}u$ and $\abs{ x^{t_0 + 1}u} =  \abs{u}  \le \alpha \ell n/2$,  we obtain 
\[\abs{u^{(t_0)}} \ge \abs{u^{(t_0 + 1)}} - \abs{u} \ge \alpha \ell n/2.\]
Finally, using that $s^{(t_0)} = (1 + x^{t_0})v$, $\alpha \ell n/2 \le \abs{u^{(t_0)}} <  \alpha \ell n$, and the~fact that $A$ is $(\alpha,\beta)$-expanding, we have:
\[ \abs{v} \ge  \frac12 \abs{s^{(t_0)}} = \frac12 \abs{A u^{(t_0)}} \ge \frac{\beta}{2}\abs{u^{(t_0)}}\ge \frac{\alpha\beta}{4} \ell n,\]
and the claim is proved. Therefore in all cases we see that $\abs{c} \ge \gamma_1\ell n$.


\textbf{Case 2.} 
In this case $u = (1+x)h$, $v = \ones{\ell}v' + Ah$; where ${v'\in \F_2^m\setminus \{\zm\}}$, ${h\in R_\ell^n}$, and $\abs{h_i} \le \ell/2$, $i\in[n]$. We~show that either ${\abs{u} \ge \gamma_2 \ell}$ or ${\abs{v} \ge \gamma_2 \ell}$, where $\gamma_2 = \frac{\alpha}{4w}\min(\beta, 1)$.  Assume the~converse, then $\abs{u} < \gamma_2\ell$ and $\abs{v} < \gamma_2\ell$. Since $Ah = v + \ones{\ell}v'$, for every $t\in [\ell]$ we get:
\[
|(1 + x^t)Ah| = |v + \ones{\ell}v' + x^t v  + x^t\ones{\ell}v'| = |v + x^t v| < 2\gamma_2\ell,
\]
where we use that $x^t \ones{\ell} = \ones{\ell}$. Further, since $A$~is~\mbox{$w$-limited}, $\abs{\ones{\ell}v'}\ge \ell$, $\abs{v} \le \gamma_2\ell\le \ell/2$, and $\alpha < 1$, we obtain: 
\[
\abs{h} \ge \frac{\abs{Ah}}{w} = \frac{\abs{v + \ones{\ell}v'}}{w}
\ge \frac{\abs{\ones{\ell}v'} - \abs{v}}{w} \ge \frac{\ell}{2w} \ge \frac{\alpha\ell}{2w}.
\]
Moreover,  if we consider $w_t = \abs{(1 + x^t)h}$, then by Lemma~\ref{lm:autcorr} there exists $t$ such that $w_t \ge \frac{\alpha\ell}{2w}$. Let us denote by $t_0$ the~smallest such~$t$.  Since $1+x^{t_0} = (1 + x) + x(1 + x^{t_0 - 1})$, $\abs{(1+x)h} = \abs{u} \le \gamma_2\ell$, and $\abs{x(1+x^{t_0 - 1})h} = w_{t_0 - 1} < \frac{\alpha\ell}{2w}$, we see that
\begin{multline*}
\abs{(1+x^{t_0})h} \le  \abs{(1+x)h} + \abs{x(1+x^{t_0 - 1})h} \\
<   \rbr{\gamma_2 + \frac{\alpha}{2w}}\ell 
\le \rbr{\frac{\alpha}{4w} + \frac{\alpha}{2w}}\ell 
= \frac{3}{4w}\alpha\ell < \alpha \ell n.  
\end{multline*}
Therefore, if we recall that the~matrix~$A$ is $(\alpha,\beta)$-expending, then from $\abs{(1+x^{t_0})h} <  \alpha \ell n$ it follows that 
\[
\abs{A(1+x^{t_0})h} \ge \beta w_{t_0} \ge \beta \frac{\alpha\ell}{2w} \ge 2\gamma_2\ell.
\]
However, we showed earlier that $\abs{A(1+x^{t})h} < 2\gamma_2\ell$ for every $t\in [\ell]$. Hence we have a~contradiction and obtain that in the second case $\abs{c} = \abs{u} + \abs{v} \ge \gamma_2\ell$.

Now, in order to finish the proof of Proposition~\ref{pr:main}, we need to set $\gamma = \min(\gamma_1,\gamma_2)$ and notice that $\dZ(\cQ) \ge \gamma \ell$ in the~both considered cases. Besides, if  $\dim \cC(A^\T(1)) = 0$ then $\im A(1) = \cC^\perp(A^\T(1)) = \F_2^m$, and by Lemma~(\ref{lm:LP-simple-cws}) we do not have case~2. Thus in this situation we obtain a~better lower bound ${\dZ(\cQ) \ge \gamma \ell n}$. 
At the same time, from (\ref{eq:LP-dual-dim}) it follows that $\dX(\cQ) = \dZ(\LP(A^\T, 1+x))$. Therefore since $A^\T$ is also $(\alpha,\beta)$-expanding,  we see, using exactly the same arguments as before, that $\dX(\cQ) \ge \gamma \ell$, and $\dX(\cQ) \ge \gamma \ell m$ in the~case when $\dim \cC(A(1)) = 0$. This completes the~proof.
\end{proof}

Now we are ready to prove our main result.

\begin{proof}[Proof of Theorem~\ref{th:main1}]
By Proposition~\ref{pr:alpha-beta} for every $\ell > 1$  there exists a~$w$-limited matrix $A\in \Mat_{m\times wn}(R_\ell)$, $n = \ceil{\gamma \ln\ell}$, $m\le \frac12 w n$, such that $A$ and $A^\T$ are $(\alpha, \beta)$-expanding, where $\alpha$, $\beta$, $\gamma$, and $w$ are some fixed constants. Consider a~quantum code~$\cQ = \LP(A, 1+x)$. It~has the~code length $N=\ell (wn+m)$, and since ${\log N \sim \log \ell}$, and ${n = \Theta(\log N)}$ as $N\to\infty$, using~(\ref{eq:lp-simple-dim}) the~dimension of~$\cQ$ is equal to ${K = \Theta(n) = \Theta(\log N)}$. 
Moreover, Proposition~\ref{pr:main} implies that $d(\cQ)\ge \gamma \ell$.  

Let us show the~upper bound $d(\cQ) \le \ell = \Omega(N/\log N)$.
Indeed, from (\ref{eq:LP-dual-dim}) it follows that $\dX(\cQ) = \dZ(\cQ')$, where $\cQ' = \LP(A^\T, b)$. If we apply  Lemma~\ref{lm:LP-simple-cws} to the~code $\cQ'$, then we obtain that $[\zm, \ones{\ell}v']$ is a~non-degenerate codeword from~$\CZ(\cQ')$ if $v' \not\in \im B^\T$, where $B = A(1)\!\in\!\Mat_{m\times wn}(\F_2)$. Since $m < wn$,  the column space of $B^\T$ does not contain some standard basis vector $e_i = (0,\dots,1,\dots,0)\in\F_2^{wn}$. Therefore $c=[\zm, \ones{\ell}e_i]$ is a non-degenerate codeword from $\CZ(\cQ')$, and $\dX(\cQ)=\dZ(\cQ') \le \abs{c} = \ell$. Hence we finally obtain that $d(\cQ)=\Theta(\ell) = \Theta(N/\log N)$, and the~proof is complete.
\end{proof}

\section{Conclusion}

We have demonstrated that the family of lifted product codes from our previous work~\cite{Panteleev&Kalachev:2019} contains QLDPC codes of dimension $\Theta(\log N)$ and \mbox{distance} $\Theta(N/\log N)$ as the~code length~$N\to\infty$. Moreover, we have shown a~way how to increase their dimension and obtained QLDPC codes of dimension~$\Omega(N^\alpha \log N)$ and \mbox{distance}~$\Omega(N^{1-\alpha/2}/\log N)$, where $0 \le \alpha < 1$. To the best of our knowledge, the~parameters of the~obtained codes are better than for previous QLDPC constructions. Let us note that the~obtained distance $\Theta(N/\log N)$ is also asymptotically  larger than the currently best-known distance $N^{1-\eps}$ for sparse subsystem codes~\cite{Bacon:2017}, where $\eps=O(1/\sqrt{\log N})$. 

We should emphasize that Proposition~\ref{pr:main} allows us, in principle, to obtain (with some extra work) QLDPC codes, where $\dZ = \Theta(N)$ and $\dX = \Theta(N/\log N)$. Though we think that we know how to achieve this, we do not present a~formal proof here in order to make our construction of the~Tanner code $\cT(G,\cC_0)$ as simple as possible.

As a~simple byproduct of the proof of Theorem~\ref{th:main1} we have also constructed a~family of classical QC~LDPC codes of any design rate and distance $\Theta(N)$ with, in some sense, optimal circulant size $\Omega(N/\log N)$.

Besides, we have further generalized our construction from~\cite{Panteleev&Kalachev:2019}, and obtained a~large class of CSS codes called lifted product codes. 
The proposed codes are quite general and contain many of the best-known quantum LDPC codes such as the hypergraph product codes~\cite{Tillich&Zemor:2009}, the bicycle codes~\cite{Mackay:2004}, the~Haah's cubic codes~\cite{Haah:2011}. 
Some of the codes from this class (e.g., the~codes $\LP(A,A^*)$ from Example~\ref{ex:LP}) have the~dimension $K=\Theta(N)$ as $N\to\infty$, but the~only upper bound on the~minimum distance we have now is $O(N/\log N)$. However, we should warn the~reader that the~methods we used in~the~proof of~Proposition~\ref{pr:main} cannot be directly applied here. Therefore it is an~interesting open problem whether some of these codes have distance that matches this upper bound. It is also interesting weather this distance can be made linear using other matrix rings $R\subseteq \Mat_{\ell\times\ell}(\F_2)$, e.g. non-commutative ones. 

We have also extended the~lifted product operation from codes to chain complexes. It naturally generalizes the~standard tensor product of two complexes, widely used in the~context of quantum codes. Though we do not discuss it here, this operation can be used to obtain quantum codes out of quantum and classical codes. For example, any quasi-cyclic (classical or quantum) code~$\cC$ of lift size $\ell$ can be combined with some other quasi-cyclic code (classical or quantum) $\cC'$ of the~same lift size in order to produce a~quantum code from the~lifted product $\cC\otimes_{\ell} \cC'$.  We think that obtaining such codes, and estimating their parameters can also be an~interesting line of future research.


\section*{Acknowledgment}
This work was supported by the Ministry of Science and Higher Education of the Russian Federation (Grant №~075-15-2020-801).

\ifCLASSOPTIONcaptionsoff
  \newpage
\fi


\bibliographystyle{IEEEtran}
\bibliography{coding,quantum}

%




\begin{IEEEbiographynophoto}{Pavel Panteleev} 
received the M.S. degree in computer science from the Moscow State Aviation Technological University in 2001 and the Ph.D. (Candidate of Science) degree in mathematics from the Moscow State University in 2006. Since 2003, he has been a~researcher at the Faculty of Mechanics and Mathematics of the Moscow State University. From 2004 to 2014 he had also been working in LSI Corporation, currently a~part of Broadcom Inc., on hardware architectures for data storage and transmission systems. Starting from 2014, he has been engaged as a~consultant in several cooperation projects with Huawei Technologies, focusing on error-correcting codes and, more recently, on quantum computing. His research interests include applications of abstract algebra and combinatorics to computer science and information theory.  
\end{IEEEbiographynophoto}

\begin{IEEEbiographynophoto}{Gleb Kalachev}
received the M.S. degree in 2013  from the Moscow State University, where he also received the Ph.D. (Candidate of Science) degree in 2018, both in mathematics.
Since 2014, he has been a~researcher at the Faculty of Mechanics and Mathematics of the Moscow State University. Starting from 2014, he also participated as a~consultant in several cooperation projects with Huawei Technologies related to data storage, error-correcting codes, and quantum computing. His research interests include  complexity theory, information theory, and cellular automata.
\end{IEEEbiographynophoto}

\vfill

\appendices

\section{Ring of circulants}\label{sc:circulants}

An~$\ell\times\ell$ circulant matrix $\mat{A}$ over $\F_2$ takes the form
\[
\mat{A} =
\begin{pmatrix}
a_{0} & a_{l-1} & \ldots & a_{1}\\
a_{1} & a_{0} & \ldots & a_{2}\\
\hdotsfor{4}\\
a_{\ell-1} & a_{\ell-2} & \ldots & a_{0}\\
\end{pmatrix},
\]
where $a_0,\ldots,a_{\ell-1}\in \F_2$. It is readily seen that the matrix $\mat{A}$ can be represented in the form
\[\mat{A} = a_0\mat{I} + a_1\mat{P} + \dots a_{\ell-1}\mat{P}^{\ell-1},\]
where $\mat{I}$ is the $\ell\times\ell$ identity matrix and
\[
\mat{P} =
\begin{pmatrix}
0 & 0 & \ldots & 1\\
1 & 0 & \ldots & 0\\
0 & 1 & \ldots & 0\\
\hdotsfor{4}\\
0 & 0 & \ldots & 0\\
\end{pmatrix}
\]
is the $\ell\times\ell$ permutation matrix representing the \emph{right} cyclic shift by \emph{one} position. Since $\mat{P}^\ell = \mat{I}$, we see that the ring of all $\ell\times\ell$ circulant matrices over $\F_2$ is isomorphic to the ring $R_\ell=\F_2[x]/(x^\ell-1)$ of polynomials over $\F_2$ modulo the polynomial $x^\ell - 1$. Hence the circulant matrix $\mat{A}$ can be uniquely represented by the polynomial:
\[a = a_0 + a_1 x + \dots + a_{\ell-1} x^{\ell-1}.\]

The algebraic structure of the ring $R_\ell$ is well studied in the~coding literature (see, e.g.,~\cite{Ling2006}). We briefly review it here.

First, let us consider the special case when $\ell$ is odd. In this case  the polynomial $x^\ell - 1$ factors into a product of  irreducible polynomials over $\F_2$
\begin{equation}\label{eq:factors-co-prime}
x^\ell - 1  = f_1(x)\cdots f_s(x).
\end{equation}
This is true, since  
\[
\gcd\left((x^\ell - 1)', x^\ell - 1\right) = \gcd\left(\ell x^{\ell - 1}, x^\ell - 1\right) = 1,
\]
and the polynomial $x^\ell - 1$ is square-free. 

In the general case we have $\ell = 2^e\ell'$, where $\ell'$ is odd. Hence it follows that 
\[x^\ell - 1 = x^{2^e\ell'} - 1 = (x^{\ell'} - 1)^{2^e}.\]
Moreover, since $\ell'$ is odd, we can apply the factorization (\ref{eq:factors-co-prime}) to the polynomial $x^{\ell'} - 1$ and obtain that 
\begin{equation}\label{eq:factors}
x^\ell - 1  = \bigl(f_1(x)\bigr)^{2^e}\cdots \bigl(f_s(x)\bigr)^{2^e}.
\end{equation}

Since the polynomials $(f_1(x)\bigr)^{2^e},\dots,(f_s(x)\bigr)^{2^e}$ are pairwise coprime, from the Chinese remainder theorem it follows that the ring $R_\ell$ is isomorphic to the direct product  
\begin{equation}\label{eq:isom-co-prime}
R^{(1)}\times \dots \times R^{(s)}
\end{equation}
of the rings $R^{(i)} = \F_2[x]/\bigl(f_i(x)\bigr)^{2^e}$, $i\in [s]$. 

When $\ell$ is odd we have $e=0$ and the rings $R^{(1)},\dots,R^{(s)}$ are in fact fields, since the polynomials $f_1(x),\dots,f_s(x)$ are irreducible over $\F_2$.

\section{Decomposition of quasi-abelian LP codes}\label{sc:LP-decomp}

Consider a~commutative group algebra $R=\F_2G$, where $\abs{G}=\ell$. Suppose that $R$ is a~direct product of rings:
\begin{equation}\label{eq:ring-decomp}
    R \cong R^{(1)} \times \dots \times R^{(s)}
\end{equation}
with the corresponding morphisms $\phi_i\colon R \to R^{(i)}$, $i\in [s]$. Let~us note that since $R$ is a~finite dimensional algebra over~$\F_2$, the~same holds for the~rings $R^{(1)},\dots,R^{(s)}$, and we have $\sum_{i\in[s]}\dim R^{(i)} = \ell$. This direct product structure implies that any~matrix $M$ over~$R$ can be uniquely represented by the~collection of matrices $\bigl(\phi_i(M)\bigr)_{i\in [s]}$, where $\phi_i(M)$ is the~matrix over $R^{(i)}$ obtained by the~action of~$\phi_i$  on each element of $M$. 

Using this idea, we can represent any code $\LP(A,B)$ constructed from matrices $A$ and $B$ over~$R$ by the~collection of $s$ codes $\bigl(\LP(A_i,B_i)\bigr)_{i\in [s]}$, where  $A_i = \phi_i(A)$, $B_i = \phi_i(B)$ are matrices over the~ring~$R^{(i)}$.
Since the direct product gives us a~one-to-one correspondence between the elements $a\in R$ and the tuples $\bigl(\phi_i(a)\bigr)_{i\in[s]}$, we also get a~one-to-one correspondence between the~codewords $c$ from $\LP(A,B)$ and the tuples of codewords $\bigl(\phi_i(c)\bigr)_{i\in[s]}$ from the~collection $\bigl(\LP(A_i,B_i)\bigr)_{i\in [s]}$. Moreover, it is not hard to check that this one-to-one correspondence also respects the degeneracy of the codewords, i.e., $c$ is degenerate iff all the~codewords $\bigl(\phi_i(c)\bigr)_{i\in[s]}$ are degenerate. This yields that:
\begin{equation*}
    \dim \LP(A,B) = \sum_{i\in [s]} \dim \LP(A_i, B_i).
\end{equation*}

In addition, if in decomposition~(\ref{eq:ring-decomp}) every ring $R^{(i)}$ is a~finite field $\F_{q_i}$, then every $\LP(A_i,B_i)$ can be uniquely represented (see Example~\ref{ex:HP-non-bin}) by a non-binary HP code defined by the~matrices $A_i$ and $B_i$ over $\F_{q_i}$, where $q_i = 2^{r_i}$, $i\in [s]$. Hence using~(\ref{eq:LP-non-bin-dim}) we obtain the following formula:  
\begin{equation}\label{eq:lp-prod-dim}
    \dim \LP(A,B) = \sum_{i\in [s]} r_i\dim \HP(A_i, B_i).
\end{equation}
At the same time, $\dim \HP(A_i, B_i)$ for each $i\in[s]$ can be found by formula~(\ref{eq:HP-dim}).

If the~lift size $\ell=\abs{G}$ is odd, then from Maschke's theorem it follows that the algebra $\F_2G$ is semisimple and hence is isomorphic~\cite[Theorem~2.4.1]{Drozd:1994} to the direct product of finite fields~${\F_{q_1}\times \dots \times \F_{q_s}}$. Hence, for matrices $A$ and $B$ over $\F_2G$, the~dimension of the~quasi-abelian code $\LP(A,B)$ is given by formula~(\ref{eq:lp-prod-dim}).

Let us also note that if $\ell$ is odd, then Lemma~\ref{lp-simple-dim} is just a special case of formula~(\ref{eq:lp-prod-dim}). Indeed, it is well known (see~Appendix~\ref{sc:circulants}) that if $\ell$ is odd, then the ring $R_\ell$ is isomorphic to the direct product of finite fields:
\[
R \cong \F_{q_1} \times \dots \times \F_{q_s},
\]
where each field $\F_{q_i}$, $q_i = 2^{r_i}$, corresponds to an~irreducible factor $f_i \in \F_2[x]$, $\deg f_i = r_i$, of the~polynomial $x^\ell - 1$. We have the following homomorphisms ${\phi_i\colon R_\ell\to \F_{q_i}}$ defined by $\phi_i\colon u\mapsto u(\beta_i)$, where $\beta_i$ is a~root of $f_i$ in the~field~$\F_{q_i}$, $i\in [s]$.  Without loss of generality we can assume that $f_1 = b$, and $\beta_1 = \beta$. Hence by (\ref{eq:lp-prod-dim}) we have\footnote{Let us note that in this sum we consider non-binary HP codes.}:
\begin{equation*}
    \dim \LP(A,b) = \sum_{i\in [s]} r_i\dim \HP(A(\beta_i), b(\beta_i)).
\end{equation*}
Since $b(\beta_i)\ne 0$ whenever $i\ne 1$, it is easy to see that $\dim \HP(A(\beta_i), b(\beta_i)) = 0$ for $i\ne 1$. At the same time, 
it~is clear that:
\begin{align*}
\dim \HP(A(\beta), 0) &= n + m - \rk [A(\beta), 0] - \rk [0, A^\T(\beta)] \\ 
                      &= \dim \cC(A(\beta)) +  \dim \cC(A^\T(\beta)),    
\end{align*}
where $n, m$ are the number of columns and rows in $A$, respectively. Thus 
\[{\LP(A, b) = r_1\!\rbr{\dim \cC(A(\beta)) +  \dim \cC(A^\T(\beta))}},\] 
where $r_1 = \deg b$, and we obtain the formula from Lemma~\ref{lp-simple-dim}.

\section{List of symbols and abbreviations}\label{sc:symbols}

\begin{tabular}{cp{0.45\textwidth}}
  $[n]$ & set $\{1,2,\dots,n\}$\\
  $\abs{u}$ &  Hamming weight of the vector~$u$\\
  $x\cap y$ & intersection of binary vectors\\
  $\F_q$ & finite field with $q$ elements\\ 
  $\mathbf{S}_n$ & set of all permutations on $[n]$\\
  $\mathbf{C}_n$ & cyclic group of order~$n$\\
  $\ell$ & lift size or circulant size\\
  $R_\ell$ & quotient ring $\F_2[x]/(x^\ell - 1)\cong \F_2\mathbf{C}_\ell$\\
  $\ones{t}$ & all one polynomial $\sum_{i=0}^{t-1} x^i$\\
  $\F G$ & group algebra over~$\F$ for the~group~$G$ \\
  $\bar{a}$ & antipode $\bar{a} = \sum_{g\in G}\alpha_g g^{-1}$ for $a\in\F G$\\
  $\cC^\perp$ & dual code for $\cC$\\
  $\cC(H)$ & code with the~parity-check matrix~$H$\\
  $\ker A$ & kernel of the linear map $v\mapsto Av$\\ 
  $\im A$ & image of the linear map $v\mapsto Av$\\ 
  $\Mat_{m\times n}(R)$ & set of all~$m\times n$ matrices over $R$\\
  $\Mat_{n}(R)$ & set of all~$n\times n$ matrices over $R$\\
  $\B(A)$ & block matrix for $A\!\in\! \Mat_{m\times n}(R)$\\
  $A^\T$ & standard transpose for $A$\\
  $A^*$ & conjugate transpose for $A$\\
  $\cQ^*$ & quantum code $\cQ$ with swapped $\CZ,\CX$\\
  $\cQ\sim \cQ'$ & permutation equivalent codes\\
  $\HP(A,B)$ & hypergraph product code\\
  $\LP(A,B)$ & lifted product code \\
  QLDPC & quantum low-density parity-check\\
  LDPC & classical low-density parity-check\\
  QC & quasi-cyclic\\
  w.h.p & with high probability
\end{tabular}

\end{document}